\documentclass{article}

\usepackage{arxiv}

\usepackage[utf8]{inputenc} 
\usepackage[T1]{fontenc}    
\usepackage{hyperref}       
\usepackage{url}            
\usepackage{booktabs}       
\usepackage{nicefrac}       
\usepackage{microtype}      
\usepackage{lipsum}
\usepackage{graphicx}
\graphicspath{ {./images/} }

\usepackage{amsmath,amsfonts}
\usepackage{algorithmic}
\usepackage{algorithm}
\usepackage{array}
\usepackage[caption=false,font=normalsize,labelfont=sf,textfont=sf]{subfig}
\usepackage{textcomp}
\usepackage{stfloats}
\usepackage{verbatim}
\usepackage{cite}
\usepackage{flushend}
\usepackage{amsthm}
\newtheorem{theorem}{Theorem}[section]
\newtheorem{proposition}{Proposition}[section]
\newtheorem{lemma}[theorem]{Lemma}
\newtheorem{problem}{Problem}

\usepackage{amsmath,amsfonts,bm}
\usepackage{xcolor}
\usepackage{fontawesome}

\def\eqref#1{(\ref{#1})}

\def\1{\bm{1}}

\def\vxi{{\bm{\xi}}}
\def\vtheta{{\bm{\theta}}}

\def\vc{{\bm{c}}}

\def\vu{{\bm{u}}}

\def\vx{{\bm{x}}}

\def\mA{{\bm{A}}}

\def\mL{{\bm{L}}}

\def\mW{{\bm{W}}}
\def\mX{{\bm{X}}}
\def\mY{{\bm{Y}}}

\DeclareMathAlphabet{\mathsfit}{\encodingdefault}{\sfdefault}{m}{sl}
\SetMathAlphabet{\mathsfit}{bold}{\encodingdefault}{\sfdefault}{bx}{n}

\def\gD{{\mathcal{D}}}
\def\gE{{\mathcal{E}}}

\def\gG{{\mathcal{G}}}

\def\gS{{\mathcal{S}}}

\def\gV{{\mathcal{V}}}

\def\sB{{\mathbb{B}}}

\def\sN{{\mathbb{N}}}

\def\sS{{\mathbb{S}}}

\newcommand{\R}{\mathbb{R}}

\DeclareMathOperator*{\argmax}{arg\,max}
\DeclareMathOperator*{\argmin}{arg\,min}

\newcommand{\abs}[1]{\left\lvert #1 \right\rvert}
\newcommand{\ip}[2]{\left\langle#1,#2\right\rangle}
\newcommand{\norm}[1]{\left\lVert#1\right\rVert}
\newcommand{\T}{\textup{T}}

\title{Graph Neural Network Based Node Deployment \\for Throughput Enhancement}

\author{
 Yifei Yang \\
  School of Electronic Information\\
  Wuhan University\\
  Wuhan, Hubei 430072 \\
  \texttt{yfyang@whu.edu.cn} \\
   \And
 Dongmian Zou \\
  Division of Natural and Applied Sciences\\
  Duke Kunshan University\\
  Kunshan, Jiangsu 215316 \\
  \texttt{dongmian.zou@dukekunshan.edu.cn} \\
  \And
 Xiaofan He \\
  School of Electronic Information\\
  Wuhan University\\
  Wuhan, Hubei 430072 \\
  \texttt{xiaofanhe@whu.edu.cn} \\
}

\begin{document}
\maketitle
\begin{abstract}
The recent rapid growth in mobile data traffic entails a pressing demand for improving the throughput of the underlying wireless communication networks.
Network node deployment has been considered as an effective approach for throughput enhancement which, however, often leads to highly non-trivial non-convex optimizations. 
Although convex approximation based solutions are considered in the literature, their approximation to the actual throughput may be loose and sometimes lead to unsatisfactory performance. 
With this consideration, in this paper, we propose a novel graph neural network (GNN) method for the network node deployment problem.
Specifically, we fit a GNN to the network throughput and use the gradients of this GNN to iteratively update the locations of the network nodes.
Besides, we show that an expressive GNN has the capacity to approximate both the function value and the gradients of a  multivariate permutation-invariant function, as a theoretic support to the proposed method.
To further improve the throughput, we also study a hybrid node deployment method based on this approach.
To train the desired GNN, we adopt a policy gradient algorithm to create datasets containing good training samples.
Numerical experiments show that the proposed methods produce competitive results compared to the baselines.
\end{abstract}


\section{Introduction}\label{sec:introduction}
The increasing number of smart mobile terminals and devices drives a massive growth in mobile data traffic, which in turn leads to a pressing need for improving the throughput of the underlying wireless communication networks \cite{chin2014emerging}.
As the wireless channel and data rate between any two network nodes are often influenced by their locations, the throughput of many wireless networks, ranging from conventional cellular networks \cite{gupta2015sinr} and vehicular communication networks \cite{wu2012cost}, to modern Internet of things communication networks \cite{you20206g} and drone-enabled communication networks \cite{merwaday2015uav}, closely hinges on the deployment of the network nodes.

In the literature, there have been some pioneering works devoted to improving the network node deployment.
Some works mainly focus on optimizing the statistics and distributions of the network node locations
\cite{wang2014coverage, sung2013attainable, lyu2020spatial, merwaday2015uav}. 
In these methods, network node locations are determined by brute force search \cite{merwaday2015uav} or random selection \cite{wang2014coverage, sung2013attainable, lyu2020spatial}, according to the obtained optimal node location statistics and distributions.
Albeit their simplicity in design, the performances of these deployment strategies are often not satisfactory.
To achieve better throughput, some other works consider the more fine-grained deployment strategy that aims to find the precise optimal location for each network node \cite{wang2018trajectory, chou2019energy, he2014dynamic, rahmati2021dynamic}.
However, the corresponding network deployment problems are much more involved and often non-convex.
To this end, approximation is commonly adopted to convert the original problem into convex ones.
For instance, an alternative expression for the throughput with respect to node locations is obtained in \cite{wang2018trajectory} by successive convex approximation. However, this alternative term is only suitable for simple scenarios with a single relay node. 
For more general scenarios with multiple nodes, the original deployment optimization problem is transformed into an equivalent Lagrange dual problem and solved by using the subgradient method in \cite{chou2019energy}. 
Another existing approach is to use the network eigenvalues derived from spectral graph theory to approximate the network throughput \cite{he2014dynamic, rahmati2021dynamic}.
Albeit the above approaches make the node deployment optimization simpler, their approximation to the actual throughput may be loose and hence sometimes lead to unsatisfactory performance.

To better address the throughput-optimal network node deployment problem, a graph neural network (GNN) based deployment method is developed in this work, which can approximate the throughput more precisely as compared to existing methods.
In particular, we model the throughput as the maximum flow (max-flow) of a network, and use a GNN to learn the functional relationship between the network max-flow and the node deployment.
Then, we derive the gradients of the network max-flow with respect to node locations by conducting backpropagation in this GNN and update the node locations accordingly.
Through in-depth analysis, we show that this method is theoretically reasonable by investigating the approximation capability of permutation-invariant neural networks, for which GNNs provide effective implementation.
In addition, a hybrid approach that combines existing spectral graph theoretic approach and the above GNN based approach is proposed to further improve the performance.

Another challenge in training the desired GNN is to gather good training data. 
Ideally, we want to collect samples of the node locations that correspond to large network throughput. 
However, since the sample space (i.e., all possible node locations) in the considered problem is prohibitively large and continuous, it is intractable to collect sufficient representative data through random sampling. To this end, we propose a reinforcement learning based approach. Specifically, we apply a GNN-based proximal policy optimization (PPO), which has sufficient exploration capabilities, to search for good representative regions with an optimal max-flow to collect good training samples.

We summarize our main contributions as follows:
\begin{enumerate}
    \item We propose a GNN-based node deployment method for network throughput enhancement. Based on this, a hybrid node deployment method is developed to further improve the throughput performance.
    \item By proving that the class of permutation-invariant neural networks has the capacity for approximating both the function values and the gradients of multivariate permutation-invariant functions, we justify the correctness and effectiveness of the proposed method.
    \item We propose to use a GNN-based PPO algorithm to collect good training data.
\end{enumerate}

The remaining parts of the article are organized as follows. Section~\ref{sec:related_work} briefly surveys related works. Section~\ref{sec:problem_maxflow} formulates the problem, and introduces a na\"{i}ve approach. In Section~\ref{sec:theoretic_bsais}, we give theoretical justification and simple numerical examples. In Section~\ref{sec:method}, we introduce the proposed methodology in detail. The experimental results are presented in Section~\ref{sec:experiment}. Eventually, Section~\ref{sec:conclusion} concludes the paper and states the limitation of our approach.

\section{Related Works}\label{sec:related_work}
\subsection{Node Deployment for Wireless Networks}\label{subsec:mol_ctrl}
Recently, network node deployment has been widely used to address the challenges of wireless communication networks, including congestion mitigation \cite{yang2018three}, coverage maximization \cite{mozaffari2016efficient}, as well as quality-of-experience  optimization \cite{chen2017caching}.
Network node deployment has also been considered for throughput enhancement. 
In \cite{merwaday2015uav}, the space is first divided into grids and then the network deployment is implemented with a brute force search method.
In \cite{wang2014coverage}, a random geometry strategy based on the Poisson point process is developed to deploy the network nodes.
In \cite{sung2013attainable} and \cite{lyu2020spatial}, uniform random deployment strategies are considered.
Despite the simplicity of these deployment strategies in design, their performances are usually limited or suffer from intractable complexity.

For precise network deployment strategies, the requirement to find the exact locations of network nodes complicates the optimization problem.
An analytical expression in \cite{wang2018trajectory} for the optimal location is derived by introducing the slack variables and leveraging the successive convex approximation.
Based on the Lagrangian method, a dual problem for the approximate optimal solution is obtained in \cite{chou2019energy} by eliminating additional variables, which can be solved by a subgradient projection method.
From a different viewpoint, the problem of maximizing the end-to-end  throughput (i.e., max-flow) under jamming is considered in \cite{he2014dynamic} and more recently in \cite{rahmati2021dynamic}.
In both of these works, the max-flow was approximated by the Cheeger constant and its maximization was done by a spectral graph theoretical approach.
Nevertheless, the approximations in these works may be loose and thus sometimes may lead to unsatisfactory performance.

\subsection{GNN in Wireless Communication}\label{subsec:rl}
The recently advocated GNN has a significant impact on wireless communications \cite{wu2020comprehensive,abadal2021computing}, as many wireless networks can be naturally abstracted into graphs.
In particular, GNNs have been widely applied to optimization problems in wireless communication networks \cite{zhang2021scalable,eisen2020optimal,guo2021learning}, given the advantage that they can effectively utilize structural and topological information. 
These existing works mainly exploit permutation invariance and equivariance of GNNs to address the challenges in wireless communication networks.
Besides, GNN can be embedded with deep reinforcement learning to enable sufficient exploration capabilities for graph-structural data in wireless communication networks.
For instance, a solution to the connection management problem based on GNN and deep Q-Learning algorithm is proposed in  \cite{orhan2021connection}, which achieved improvement in several objectives.
In \cite{moon2021neuro}, a multi-agent PPO algorithm is proposed, which introduces the graph convolutional network with an empirically driven approach to training wireless MAC controllers.
However, the problems considered in the above works are different from our research problem, and hence their solutions may not be directly applicable.

\subsection{Approximation using GNN}\label{subsec:gnn}
GNNs have been applied to representation learning on graphs since the seminal work of \cite{bruna2013spectral}. The important properties of permutation equivariance and invariance make GNNs suitable for graph data since they do not depend on a specific ordering of nodes \cite{kondor2018covariant, zou2020graph, keriven2019universal, maron2018invariant}. 
However, due to these invariance properties, GNNs are not expected to enjoy universal approximation \cite{hornik1991approximation}. In fact, GNNs are not sufficiently expressive for representing all permutation-invariant neural networks, and the equivalence between expressivity and approximation capability has been established in \cite{chen2019equivalence}. In particular, the Weisfeiler-Leman test is an important benchmark for the expressivity of GNNs and recent studies have focused on going beyond such limits \cite{xu2018how, morris2019weisfeiler, bodnar2021weisfeiler}. 
The universality of a specific class of invariant/equivariant GNNs is proved in \cite{keriven2019universal}, relying on a new generalized Stone-Weierstrass theorem for the algebra of real-valued equivariant functions. The more general notion of universal approximation of group-invariant networks is studied in \cite{maron2019universality}. 

Many applications use the gradients of neural networks to approximate those of the original functions. An early study  numerically studies training neural networks to approximate a multivariate function and its first and second partial derivatives \cite{nguyen1999approximation}. When both function and gradient values are available, it is also possible to consider Sobolev training where the loss function is inherited from the Sobolev spaces \cite{czarnecki2017sobolev, vlassis2021sobolev, son2021sobolev}. However, Sobolev training is not applicable to our problem, since it is difficult to obtain precise gradients in our dataset.

\section{Problem Formulation and a Na\"{i}ve Solution}\label{sec:problem_maxflow}
\subsection{Problem Formulation}\label{subsec:problem_formulation}
In this work, we consider a wireless communication network with $n$ nodes, as shown in Fig.~\ref{network}. The node deployment is denoted by $\left \{ \vxi_{i}= \left ( x_{i}^{\left ( 1 \right ) } , x_{i}^{\left ( 2 \right ) }\right ) \in \mathbb{R}^{2}  : i= 1,\dots ,n  \right \}$, where $\vxi_1$ and $\vxi_n$ are the locations of a pair of stationary source and destination nodes, and $\{\vxi_i\}_{i=2}^{n-1}$ are the locations of mobile relay nodes. With a slight abuse of notation, when there is no confusion, we may also use $\{\vxi_i\}_{i=1}^{n}$ to denote the nodes themselves.
We model this wireless communication network using a simple graph $\gG$ specified as follows. The vertex set of $\gG$, denoted by $\gV$, consists of all nodes $\{\vxi_i\}_{i=1}^n$. The edge set of $\gG$, denoted by $\gE$, consists of all pairs of nodes. 
We assume that for a pair of nodes $\vxi_i$ and $\vxi_j$, the corresponding communication rate between the pair of nodes is a function of their locations and is denoted by $f_{\rm cap}(\vxi_i, \vxi_j)$.
Moreover, we assume $f_{\rm cap} (\vxi_i, \vxi_j)$ is a Lipschitz continuous function with respect to $\vxi_i$ and $\vxi_j$, which holds in most wireless scenarios.
Accordingly, we construct the weighted adjacency matrix $\mA$ of this network, whose $(i,j)$-th entry $ A_{i,j}$ is defined as the communication rate between nodes $i$ and $j$. That is, 
\begin{equation}
  \label{aij}
  A_{i,j} \overset{\bigtriangleup }{= }
  \left\{\begin{matrix}
  f_{\rm cap} (\vxi_i, \vxi_j)  &, ~i\ne j \\
  0 &, ~i= j \end{matrix}\right. .
\end{equation}

Our objective is to find the optimal locations of the relay nodes so that throughput (i.e., max-flow) between the source node and the destination node is optimized. 
It is well-known \cite{ford1956maximal} that when the locations of the nodes are specified, this max-flow equals to the min-cut $C$ of the underlying network graph, defined by
\begin{equation}
  \label{Csgl}
  C = \min_{\{S:\vxi_{1}\in S,\vxi_{n}\in \bar{S} \} } \sum_{\vxi_{i}\in S,\vxi_{j}\in \bar{S}} A_{i,j} ,
\end{equation}
where $S$ is a subset of $\gV$, $\bar{S}$ is the complement of $S$, the minimum is taken over all partitions $\{S, \bar{S}\}$ of $\gV$. 
\begin{figure}[!t]
  \centering
  \includegraphics[width=2.5in]{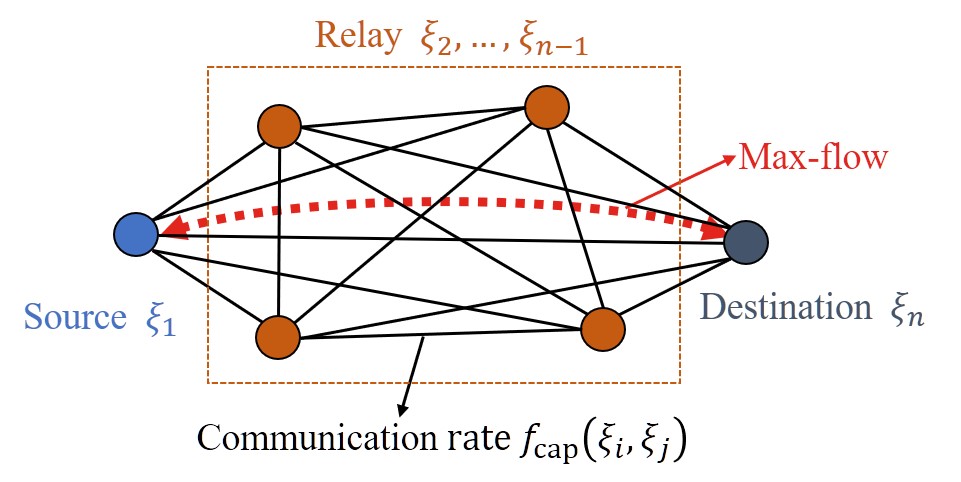}
  \caption{A wireless communication network with a single source to a single destination.}
  \label{network}
\end{figure}

Clearly, the maximum network flow $C$ is a function of the edge weights $A_{i,j}$'s, and thus also a function of the locations $\vxi_i$ of the nodes. 
Actually, $C$ is a $1$-Lipschitz function with respect to each edge weight $A_{i,j}$. We state this fact as a lemma.
\begin{lemma}\label{lemma:1-lip}
$C$ is a $1$-Lipschitz function of $A_{i,j}$ for each $i,j \in \{1,\dots,n\}$.
\end{lemma}
\begin{proof}
This simply follows from the fact that the change of $C$ cannot exceed the perturbation of $A_{i,j}$ because of the summation in \eqref{Csgl}. More specifically, we proceed as follows. Fix two different indices $i_0$ and $j_0$. For the adjacency matrix $\mA$, let $S^*$ and $\overline{S^*}$ be a partition of $\gV$ with which the minimum in \eqref{Csgl} is achieved, that is, $C = \sum_{\vxi_{i}\in S^*,\vxi_{j}\in \overline{S^*}} A_{i,j}$. Consider the perturbed adjacency matrix $\tilde{\mA}$ where $\tilde{A}_{i,j} = A_{i,j}$ if $(i,j) \neq (i_0, j_0)$; and $\tilde{A}_{i_0, j_0} \neq A_{i_0, j_0}$. Let $\tilde{C}$ denote the max-flow corresponding to $\tilde{\mA}$. If $\tilde{A}_{i_0, j_0} > A_{i_0, j_0}$, it is clear that
\begin{equation}
    \begin{aligned}
     \tilde{C} & = \min_{ \{ S: \vxi_{1}\in S,\vxi_{n}\in \bar{S} \} } \sum_{\vxi_{i}\in S,\vxi_{j}\in \bar{S}} \tilde{A}_{i,j} \\
     & \leq \sum_{\vxi_{i}\in S^*,\vxi_{j}\in \overline{S^*}} \tilde{A}_{i,j},
    \end{aligned}
\end{equation}
and thus
\begin{equation}
    \begin{aligned}
     & \quad \abs{\tilde{C} - C} = \tilde{C} - C \\
     & \leq \sum_{\vxi_{i}\in S^*,\vxi_{j}\in \overline{S^*}} \tilde{A}_{i,j} - \sum_{\vxi_{i}\in S^*,\vxi_{j}\in \overline{S^*}} A_{i,j} \\
     & = \tilde{A}_{i_0, j_0} - A_{i_0, j_0} = \abs{\tilde{A}_{i_0, j_0} - A_{i_0, j_0}}.
    \end{aligned}
\end{equation}
On the other hand, if $\tilde{A}_{i_0, j_0} < A_{i_0, j_0}$, then we consider $\mA$ as a perturbed version of $\tilde{\mA}$. For the same reason as above, 
\begin{equation}
    \abs{\tilde{C} - C} \leq \abs{A_{i_0, j_0} - \tilde{A}_{i_0, j_0}}.
\end{equation}
Combining both cases, we conclude that $C$ is a $1$-Lipschitz function of $A_{i,j}$. 
\end{proof}

As a composition of Lipschitz functions, $C$ itself is Lipschitz with respect to the locations $\vxi_i$ of the nodes as well (see e.g., \cite[Lemma 2.1]{kim2021lipschitz}). Accordingly, Rademacher's Theorem \cite[Theorem 3.1.6]{federer2014geometric} immediately yields the following result.
\begin{lemma}
\label{lemma_C}
For $i = 2, \dots, n-1$, $C$ is a differentiable function of $\vxi_i$ almost everywhere.
\end{lemma}

With the above notations, to optimize the throughput, we need to resolve the following problem.
\begin{problem}
\label{problem_maxflow}
For fixed source $\vxi_1$ and destination $\vxi_n$,
find
\begin{equation}\label{eq:max_max_flow}
    \vxi_2^*, \cdots, \vxi_{n-1}^* = \argmax C,
\end{equation}
where the maximum is taken over a designated domain $\gD$ (i.e. the deployable region of the relay nodes) of $\vxi_2, \cdots, \vxi_{n-1}$.
\end{problem}
We remark that this problem is permutation invariant. That is, if $(\vxi_2^*, \cdots, \vxi_{n-1}^*)$ is an optimizer of \eqref{eq:max_max_flow}, then its any permuted version 
$(\vxi_{\Upsilon(2)}^*, \cdots, \vxi_{\Upsilon(n-1)}^*)$, where $\Upsilon$ is a bijection so that $\{\Upsilon(2), \dots, \Upsilon(n-1)\} = \{2, \dots, n-1\}$, is also an optimizer.

To solve Problem~\ref{problem_maxflow}, we employ the following iterative strategy.
Start with certain positions $\vxi_2[0], \cdots, \vxi_{n-1}[0]$, and update the relay nodes along ascending directions (e.g., the directions of gradients) of the objective function $C$ with a small stepsize $\zeta$. The node locations after $\mathfrak{t}$ updates are denoted by $\vxi_2[\mathfrak{t}], \cdots, \vxi_{n-1}[\mathfrak{t}] $, for $\mathfrak{t} = 1, \dots, \mathfrak{T}$, where $\mathfrak{T}$ is the step count. These directions will be provided by a GNN as illustrated below.

To construct the corresponding GNN, we consider the graph $\gG = \gG(\gV, \gE, \mA)$ defined in Section \ref{sec:problem_maxflow}. 
For each node, we define its features to be a three-dimensional vector $\vx_i = (x_i^{(0)}, x_i^{(1)}, x_i^{(2)}), i = 1,\dots,n$. 
Here, the first coordinate $x_i^{(0)}=1$ if the node is either the source or the destination (i.e., when $i=1$ or $i=n$) and $x_i^{(0)}=0$ otherwise. 
The coordinates $(x_i^{(1)}, x_i^{(2)})$ are exactly the location $\vxi_i$ of the $i$-th node. 
We define an $n \times 3$ matrix $\mX$, whose $i$-th row is given by $\vx_i$, as a compact representation of the node features.
With the above notations, all the information of the wireless communication network is abstracted into the attributed graph $\gG$. We use $\gG$ and $\mX$ as the inputs of the GNN. 
Throughout this paper, we use $g$, a function of both $\mX$ and $\mA$, to denote the functional expression of a GNN. 

\subsection{Na\"{i}ve Approach: Gradient Learning (GL) with GNN}\label{subsec:naive_approach}
Since the gradients of $C$ with respect to the locations $\vxi_2, \dots, \vxi_{n-1}$ exist almost everywhere by Lemma~\ref{lemma_C} and provide a direction of increasing $C$, a na\"{i}ve approach is to fit these gradients with a GNN and then update the relay nodes according to the output (i.e. the estimation of the gradients) provided by this GNN.
However, the desired training data of this GNN are not readily available, as it is quite challenging to directly compute the gradients of $C$ based on the combinatorial expression \eqref{Csgl}.
For this reason, instead of finding the exact gradients, we train this GNN to find unit vectors $\vu_2, \dots, \vu_{n-1}$ that lead to increment of $C$ and update the node locations from $(\vxi_2, \dots, \vxi_{n-1})$ to $(\vxi_2+\zeta \vu_2, \dots, \vxi_{n-1}+\zeta \vu_{n-1})$ at each iteration. 
As discussed in Section \ref{subsec:dataset_creation}, these unit vectors can be found via reinforcement learning with sufficient exploration capabilities.
Therefore, the problem of learning gradient with GNN can be treated as a node-level regression task, to be solved by the GNN $g$. 
Specifically, this GNN outputs an $(n-2) \times 2$ matrix $g(\mX, \mA)$, whose rows are normalized to obtain the aforementioned unit vectors as follows:
\begin{equation}
\displaystyle \mY = [\vu_{i}]_{i=2}^{n-1} = \left[\frac{g_{i}(\mX, \mA)}{\norm{g_{i}(\mX, \mA)}}\right]_{i=2}^{n-1},
\end{equation}
where $\vu_i$ is the $(i-1)$-th row of the $(n-2) \times 2$ matrix $\mY$ and $g_i(\mX, \mA)$ is the $(i-1)$-th row of $g(\mX, \mA)$, respectively. 
We supervise the training of GNN by minimizing the mean squared error (MSE) $\norm{\mY - \tilde{\mY}}_{\rm F}^2$, where $\tilde{\mY}$ is the matrix of the unit vectors provided by the training data obtained from reinforcement learning. 
After this GNN is trained, in the test phase, at each step $\mathfrak{t} = 0,\dots,\mathfrak{T}-1$, the relay node locations are updated as follows
\begin{equation}\label{eq:update_location}
    \vxi_i[\mathfrak{t}+1] = \vxi_i[\mathfrak{t}] + \zeta \vu_{i-1}[\mathfrak{t}], i = 2,\cdots, n-1.
\end{equation}
In \eqref{eq:update_location}, $\vu_{i-1}[\mathfrak{t}]$ is the ($i-1$)-th row of $\mY[\mathfrak{t}]$, which is the normalized version of $g(\mX[\mathfrak{t}], \mA[\mathfrak{t}])$.
An overview of this procedure is shown in Fig.~\ref{GNN_gradient_framework}.

\begin{figure}[!t]
  \centering
  \includegraphics[width=\columnwidth]{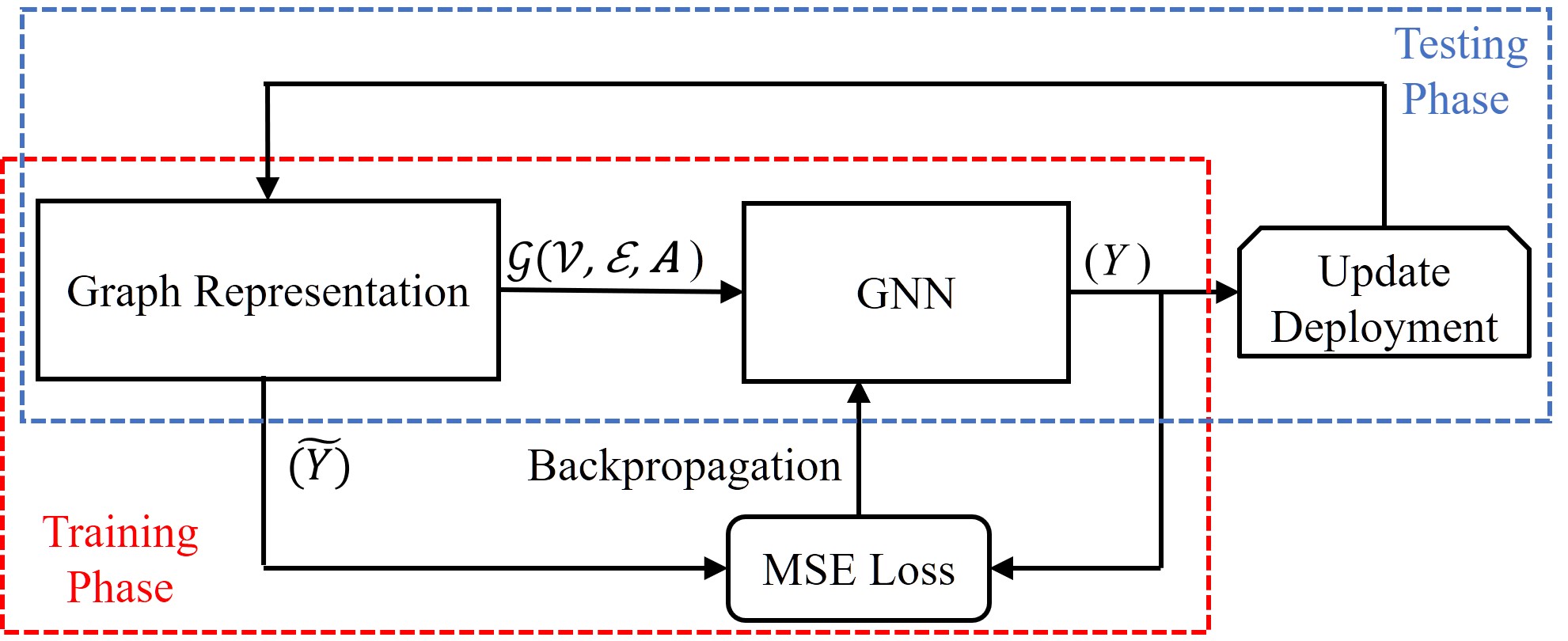}
  \caption{The framework of gradient learning. The red box illustrates the training phase, which learns the gradients with GNN; the blue box illustrates the testing phase, which updates the relay nodes to maximize the max-flow.}
  \label{GNN_gradient_framework}
\end{figure}

\section{Theoretic Basis}\label{sec:theoretic_bsais}
As will be shown in Section \ref{subsec:compar_GL}, the na\"{i}ve approach still cannot deliver satisfactory performance. 
To address this issue, we will describe an alternative approach later in Section \ref{sec:method}.
In particular, note that the value of $C$ can be efficiently computed by the Ford-Fulkerson algorithm \cite{ford1956maximal}. 
This motivates us to directly estimate the value of $C$ by a GNN, and then we use the gradients of this GNN to approximate the gradients of $C$ and update the relay nodes accordingly. 

Before presenting the details, some theoretic analysis is provided to validate the feasibility of this alternative approach.
Since $C$ is permutation invariant with respect to the node indices, we can achieve an arbitrarily small error of approximating $C$ using an expressive GNN, due to its capability of universal approximation \cite{chen2019equivalence}, which is a variant of the multi-layer perceptron (MLP) version \cite{hornik1989multilayer} modulo equivariance.\footnote{Even if we have a good approximation of the function, its gradient is not necessarily close to that of the original function in general. For instance, let $f(x) = c$ and $\tilde{f}(x) = c + \upsilon \sin(\iota x)$ where $c$ is a constant. Clearly, $\norm{f-\tilde{f}}_{\infty} = \upsilon$ while $\norm{f'-\tilde{f}'}_{\infty} = \upsilon \iota$. If we take a small $\upsilon$ but a large $\iota$, the gap between the function values will be small, but the norm of the difference between the gradients can be arbitrarily large.}

\subsection{Theoretic Analysis}\label{subsec:approx_both_func_grad}
We first state the following proposition on both function and gradient approximations for permutation-invariant ReLU neural networks.

\begin{proposition}\label{thm:approx}
Let $f: [0,1]^n \rightarrow \R$ be a differentiable permutation-invariant function such that $\nabla f$ is $L$-Lipschitz. Then for any $\epsilon, \delta > 0$, there exists an permutation-invariant ReLU neural network $g: [0,1]^n \rightarrow \R$ for which both $\abs{f(\vx) - g(\vx)} \leq \epsilon$ and $\norm{\nabla f(\vx) - \nabla g(\vx)} \leq \epsilon$ hold for any $\vx \in [0,1]^n-\sS$ where $\sS$ is a subset of $[0,1]^n$ with Lebesgue measure $m(\sS)<\delta$.
\end{proposition}

\begin{proof}
We first construct a continuous piecewise linear function that approximates $f$. For some $N \in \sN$ that we will determine later, divide $[0,1]^n$ into $N^n$ boxes of equal size, denoted by $\sB_{i_1, \dots, i_n}, i_1, \dots, i_n \in \{1,\dots,N\}$, where
\begin{align*}
    \sB_{i_1, \dots, i_n} \overset{\bigtriangleup }{= } \Bigg\{ \vx \in [0,1]^n: \frac{i_l-1}{N} \leq x_l \leq \frac{i_l}{N}, l = 1, \dots, n \Bigg\}.
\end{align*}
We also define $\tilde{\sB}_{i_1, \dots, i_n}$, for $i_1, \dots, i_n \in \{1,\dots,N\}$ by
\begin{align*}
    \tilde{\sB}_{i_1, \dots, i_n} \overset{\bigtriangleup }{= } \Bigg\{ \vx \in [0,1]^n: \frac{i_l-1+\omega}{N} \leq x_l & \leq \frac{i_l-\omega}{N}, \\ l &= 1, \dots, n \Bigg\},
\end{align*}
for some $\omega > 0$ that we determine later. We denote the centers of the boxes $\sB_{i_1, \dots, i_n}$, which at the same time are the centers of $\tilde{\sB}_{i_1, \dots, i_n}$, to be $\displaystyle \vc_{i_1, \dots, i_n} = \left( \frac{i_1-1/2}{N}, \dots, \frac{i_n-1/2}{N} \right)$. 

Let $\tilde{h}: [0,1]^n \rightarrow \R$ be defined by 
\begin{equation}
    \tilde{h}(\vx) = 
    \begin{cases}
        \tilde{h}_{i_1, \dots, i_n}(\vx) & ,~\text{if}~ \vx \in \tilde{\sB}_{i_1, \dots, i_n} \\
        0 & ,~\text{otherwise}
    \end{cases}
    ,
\end{equation}
where 
\begin{equation}
    \tilde{h}_{i_1, \dots, i_n}(\vx) = f(\vc_{i_1, \dots, i_n}) +  \nabla f(\vc_{i_1, \dots, i_n})^\T (\vx - \vc_{i_1, \dots, i_n}).
\end{equation}
Clearly $\tilde{h}$ is linear on $\tilde{\sB}_{i_1,\dots,i_n}$ and permutation-invariant. Moreover, there exists a piecewise linear function $h$ such that $h = \tilde{h}$ on $\tilde{\sB}_{i_1,\dots,i_n}$, for instance, by interpolating the vertex values using nodal basis functions \cite[Ch.~3]{brenner2008mathematical} on $\sB_{i_1,\dots,i_n} - \tilde{\sB}_{i_1,\dots,i_n}$. Although $h$ is not necessarily permutation invariant, we construct a function $g: [0,1]^n \rightarrow \R$ such that
\begin{equation}
    g(\vx) = \frac{1}{n!} \sum_{\vx^\dagger \in \gS(\vx)} h(\vx^\dagger),
\end{equation}
where $\gS(\vx)$ is the set of all permutations of $\vx$. Clearly, $g$ is permutation-invariant and its value agrees with $h$ on each $\tilde{\sB}_{i_1,\dots,i_n}$ since $h$ is permutation-invariant on $\tilde{\sB}_{i_1,\dots,i_n}$. Moreover, on each $\tilde{\sB}_{i_1,\dots,i_n}$, we have on one hand
\begin{equation}
\begin{aligned}
    \norm{f(\vx) - g(\vx)} &\leq \norm{\nabla f} \norm{\vx - \vc_{i_1,\dots,i_n}} \\
    &\leq M \frac{\sqrt{n}}{N},
\end{aligned}
\end{equation}
where $M$ is a uniform upper bound of $\nabla f$ on $[0,1]^n$, which is finite due to the compactness of $[0,1]^n$.
On the other hand,
\begin{equation}
\begin{aligned}
    \norm{\nabla f(\vx) - \nabla g(\vx)} &= \norm{\nabla f(\vx) - \nabla f(\vc_{i_1,\dots,i_n})} \\
    &\leq L \norm{\vx - \vc_{i_1,\dots,i_n}} \leq L \frac{\sqrt{n}}{N}.
\end{aligned}
\end{equation}
Given $\epsilon, \delta > 0$, we choose $\displaystyle N = \left\lceil \frac{\epsilon}{\min\{M, L\} \sqrt{n}} \right\rceil + 1$. We also choose $\omega$ so small that $m(\sS) < \delta$ where 
\begin{equation}
\label{eq:s}
    \sS = [0,1]^n - \bigcup_{i_1, \dots, i_n = 1}^N \tilde{\sB}_{i_1,\dots,i_n}.
\end{equation}
Lastly, by \cite[Theorem 5.2]{he2020relu}, $g$ can be implemented by a ReLU neural network with $\lceil \log_2 (n+1) \rceil$ layers.
\end{proof}

We remark that the domain $[0,1]^n$ of $f$ can be replaced by any compact permutation-invariant domain, which is the case in the considered node deployment problem.

Proposition~\ref{thm:approx} indicates that there exists a GNN such that, the GNN value is a close approximation of the function value of $C$ in \eqref{Csgl}, and at the same time, the gradients of this GNN well approximate the gradient of $C$ since
GNN can be a practical implementation of the permutation invariant neural network $g$. 
In GNN, permutation-invariance is guaranteed by using graph convolutional layers, also known as message passing layers.
In general, a graph convolutional layer can be represented as
\begin{equation}\label{eq:graph_conv_layer}
    \vx_i^{\prime} = \gamma_{\vtheta} \left( \vx_i,
    {\rm AGG}_{j \in \mathcal{N}(i)} \, \phi_{\vtheta}
    \left(\vx_i, \vx_j, A_{i,j}\right) \right),
\end{equation}
where $\mathcal{N}(i)$ denotes the neighbors of $i$, ${\rm AGG}$ is a permutation invariant aggregation function (e.g., sum), and $\gamma_{\vtheta}$ and $\phi_{\vtheta}$ denote differentiable functions such as fully-connected networks. 
Graph convolutional layers in the form of \eqref{eq:graph_conv_layer} all enjoy permutation equivariance, but have different expressivity, and thus different approximation power.
It is important to take a more expressive GNN to guarantee that a wide range of permutation-invariant neural networks can be represented, as studied in \cite{chen2019equivalence, azizian2021expressive, jegelka2022theory}. 
To this end, we adopt the graph convolutional layer from \cite{morris2019weisfeiler} (denoted as GraphConv), which is a higher-order GNN that considers multi-scale graph structures and proves more expressive than the Weisfeiler-Leman test. Specifically, the message passing operation in GraphConv is represented as
\begin{equation}
\label{graph_conv}
    \vx^{\prime}_i = \mW_1 \vx_i + \mW_2
\sum_{j \in \mathcal{N}(i)} A_{i,j} \cdot \vx_j,
\end{equation}
where $ \mW_{1} $ and $ \mW_{2} $ are matrices whose entries are trainable.
We will also compare the performance with a less expressive GNN in ablation studies to validate this choice.
On the other hand, to facilitate differentiation, especially for approximation of the gradients, we use the Gaussian Error Linear Unit (GELU) \cite{hendrycks2016gaussian} as the nonlinear activation function instead of ReLU in our experiments. We remark that GELU is very close to ReLU in value but much smoother.

\subsection{Numerical example}\label{subsec:numerical_example}
In this section, we will show synthetic numerical examples in which we use GNN to approximate multivariate functions and study their gradients. 

We consider a simple fully-connected graph with three nodes, associated with an input features matrix $\mX \in \R^{3 \times 2}$ and a scalar output $y \in \R$. We assume that the true value of $y$ is an analytical function of $\mX$. Specifically, we consider the following two functions
\begin{equation}
\begin{aligned}
    f_1(\mX) &= x_{11}^{2} +  x_{12}^{2}  +  x_{21}^{2} +  x_{22}^{2} +  x_{31}^{2} +  x_{32}^{2},\\
    & \qquad \qquad x_{11},x_{12},x_{21},x_{22},x_{31},x_{32} \in (1,4);
\end{aligned}
\end{equation}
and
\begin{equation}
\begin{aligned}
    f_2(\mX) &= x_{11} \cdot  x_{12}  +  x_{21} \cdot  x_{22} +  x_{31} \cdot  x_{32},\\
    & \qquad \qquad x_{11},x_{12},x_{21},x_{22},x_{31},x_{32} \in (1,4).
\end{aligned}
\end{equation}

For each $f_i$, $i=1,2$, we take a collection of $30,000$ training examples of pairs $(\mX, f_i(\mX))$. In each training example, we randomly sample $x_{11},x_{12},x_{21},x_{22},x_{31},x_{32}$ independently from the uniform distribution ${\rm Unif}(1,4)$ and then compute $y = f_i(\mX)$. We also collect $500$ testing examples in the same manner. We then train a GNN with a cascading of the following layers:
\begin{itemize}
    \item Input layer: 2-dimensional;
    \item $3 \times$ GraphConv layers: 32-dimensional;
    \item Global addition pooling (global-add-pool);
    \item $2 \times$ Linear layers: 32-dimensional;
    \item Output layer: 1-dimensional.
\end{itemize}
Between the GraphConv layers, we use the GELU activation function and the GraphSizeNorm \cite{dwivedi2020benchmarking}
where we apply normalization over each individual graph in a batch of node features.
Moreover, the activation function between the linear layers is also GELU.
For each setting, we fit the GNN to the function value with a learning rate of $0.002$, and a batch size of $100$, for $1,000$ epochs.

Suppose the trained GNN is denoted by $g$. For each testing example $\mX$, we compute both $g(\mX)$, the output value of GNN, and $\nabla g(\mX)$, the gradients of GNN with respect to the input variables. To evaluate the performance of GNN, we use the maximum relative error, defined as
\begin{equation}
\begin{aligned}
\max_\mX \frac{\abs{g(\mX) - f_i(\mX)}}{\abs{f_i(\mX)}} 
\end{aligned}
\end{equation}
for the function value and 
\begin{equation}
\begin{aligned}
\max_\mX \frac{\abs{\partial_{x_{kl}} g(\mX) - \partial_{x_{kl}} f_i(\mX)}}{\abs{\partial_{x_{kl}} f_i(\mX)}}, \quad k, l \in \{1,2,3\}
\end{aligned}
\end{equation}
for the gradients, where the maximum is taken over the collection of the testing examples in each case.

\begin{table}[!t]
  \caption{The maximum relative errors between GNN and function outputs}
  \label{maximum_relative_error}
  \centering
  \renewcommand\arraystretch{1.25}
  \setlength{\tabcolsep}{0.2in}{
  \begin{tabular}{|c||c|c|}
    \hline
    ~                                      & \multicolumn{2}{c!{\vrule}}{maximum relative error}  \\ 
    \hline
    $i$                               & $1$ & $2$                                          \\ 
    \hline
    $f_i$                     & 0.11\%     & 0.24\%                                               \\ 
    \hline
    $\partial_{x_{11}}f_i$  & 1.32\%     & 2.15\%                                               \\ 
    \hline
    $\partial_{x_{12}}f_i$ & 2.82\%     & 1.58\%                                               \\ 
    \hline
    $\partial_{x_{21}}f_i$ & 3.04\%     & 1.83\%                                               \\ 
    \hline
    $\partial_{x_{22}}f_i$ & 6.15\%     & 1.74\%                                               \\ 
    \hline
    $\partial_{x_{31}}f_i$ & 3.76\%     & 6.03\%                                               \\ 
    \hline
    $\partial_{x_{32}}f_i$ & 2.07\%     & 5.89\%                                               \\
    \hline
\end{tabular}}
\end{table}

\begin{figure}[!t]
  \centering
  \subfloat[]{\includegraphics[width=1.75in]{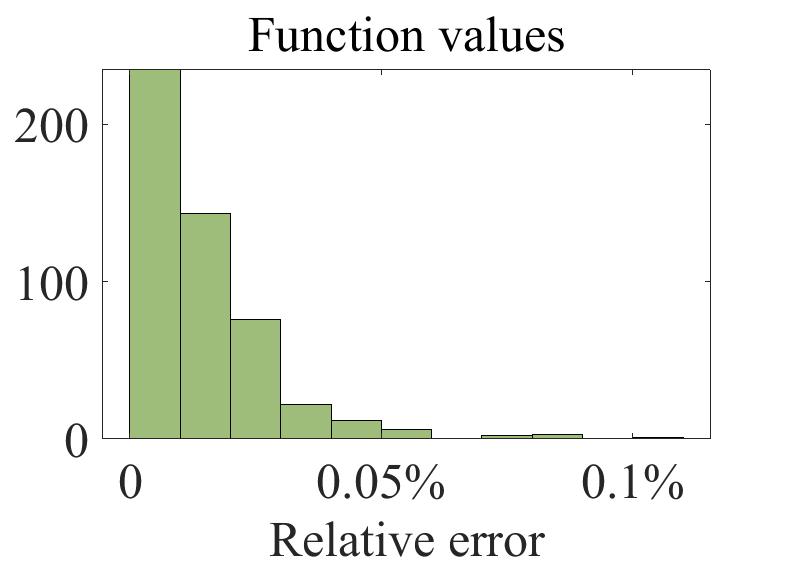}%
  \label{function1_value_error}}
  \subfloat[]{\includegraphics[width=1.75in]{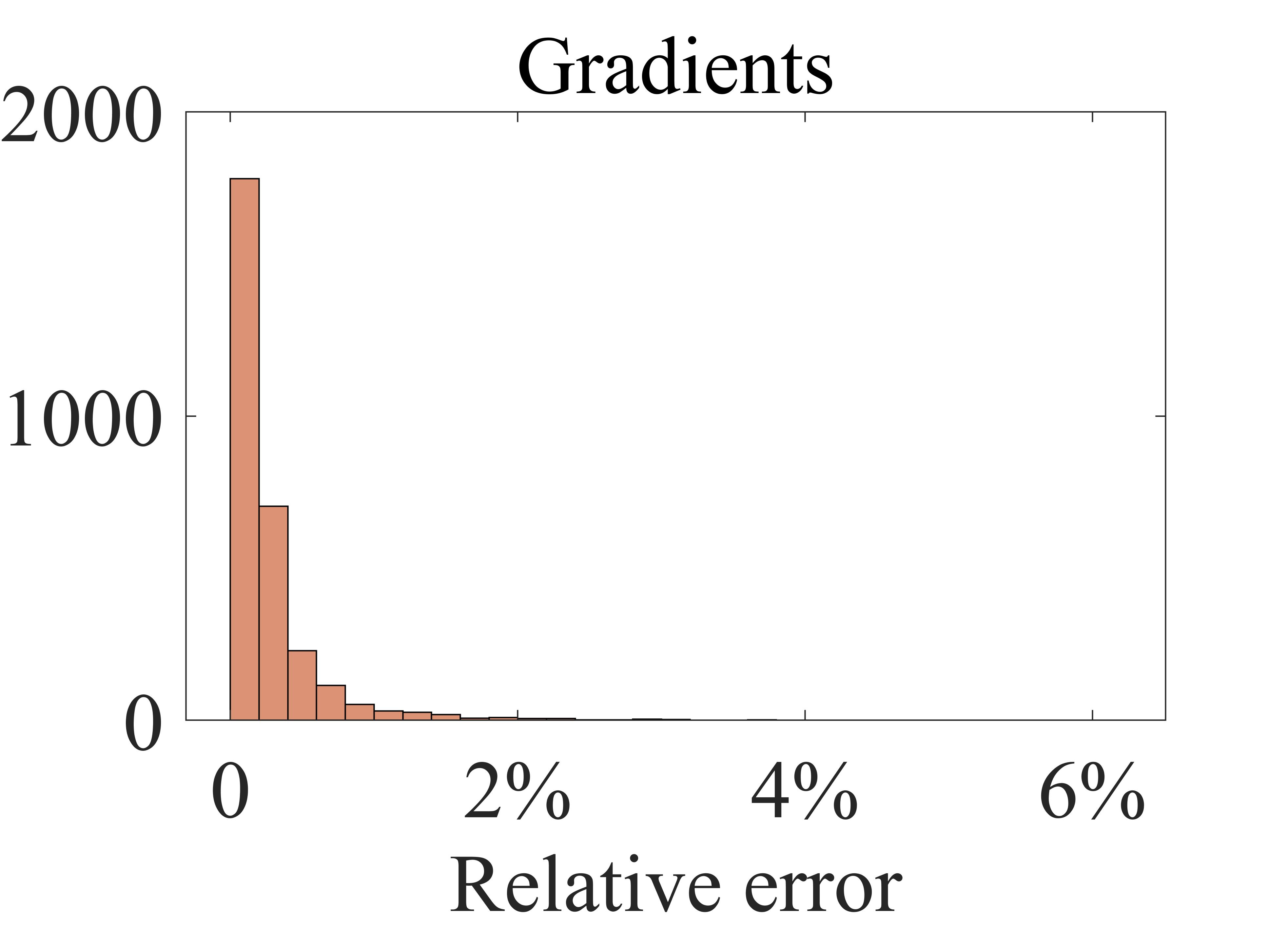}%
  \label{function1_gradient_error}}
  \hfil
  \subfloat[]{\includegraphics[width=1.75in]{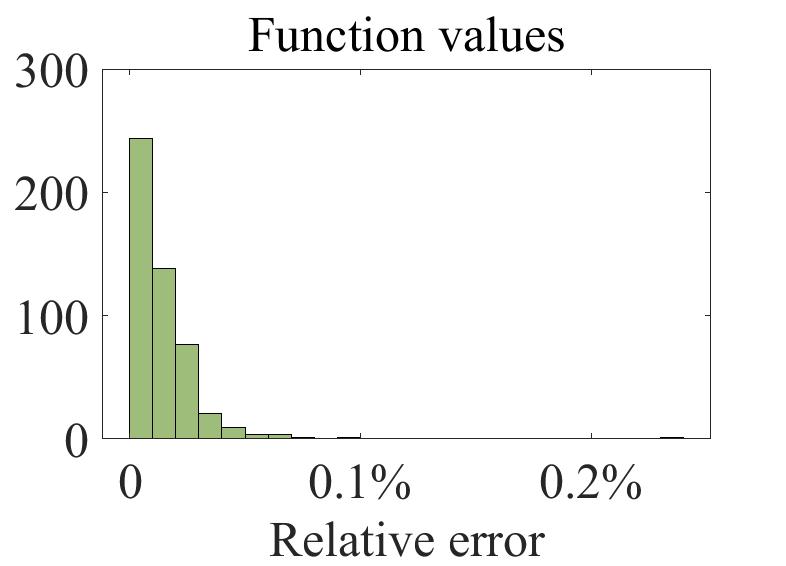}%
  \label{function2_value_error}}
  \subfloat[]{\includegraphics[width=1.75in]{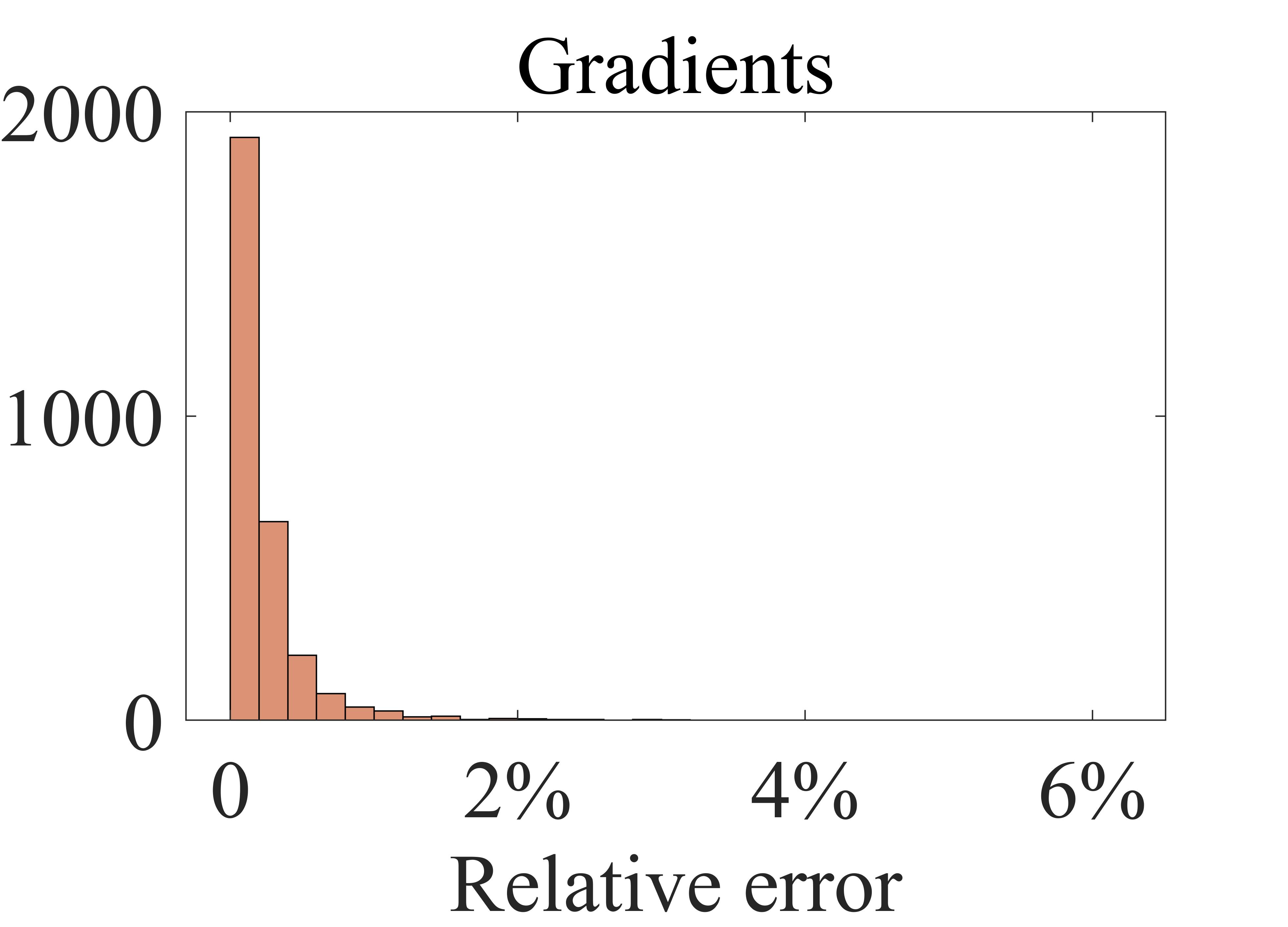}%
  \label{function2_gradient_error}}
  \caption{The distribution of the relative errors of function values and first-order partial derivatives. (a) Relative errors for $f_1$. (b) Relative errors for $\partial_{x_{kl}}f_1$ (collection of all $k,l$'s). (c) Relative errors for $f_2$. (d) Relative errors for $\partial_{x_{kl}}f_2$ (collection of all $k,l$'s).}
  \label{function}
\end{figure}

The experimental results are shown in Table~\ref{maximum_relative_error}, where we observe consistently small errors from the partial derivatives, albeit not included in the loss functions. For a more detailed discussion, we also show the distributions of the relative errors of function value and first-order partial derivatives in Fig.~\ref{function}.
For $f_1$, the maximum relative error is $0.11\%$, but most of the relative errors of function value are within $0.05\%$. 
The maximum relative error of the first-order partial derivatives is $6.15\%$, but $99.03\%$ of them are within $2\%$.
On the other hand, for $f_2$, the maximum relative error is $0.24\%$, but $96\%$ of the relative errors of function value are within $0.04\%$. The maximum relative error of the first-order partial derivatives is $6.03\%$, but $98.77\%$ of them are within $1.5\%$.
From these experimental results, we see that after fitting GNN to the function value, it can not only faithfully represent a multi-variable function but also well approximate its first-order partial derivatives. 

\section{Proposed Method}\label{sec:method}
\subsection{Max-Flow Learning (MFL) with GNN}\label{subsec:learn_maxflow}

Motivated by the above theoretic result, we propose to directly fit a GNN to the max-flow of the network and determine the update directions of the relay nodes by differentiating the GNN, as shown in Fig.~\ref{GNN_value_framework}.
Unlike the na\"{i}ve method of learning gradients, the learning objective of this approach is the scalar max-flow value of the network, which replaces the original node-level task to a graph-level task. 
\begin{figure*}[!t]
  \centering
  \includegraphics[width=.8\textwidth]{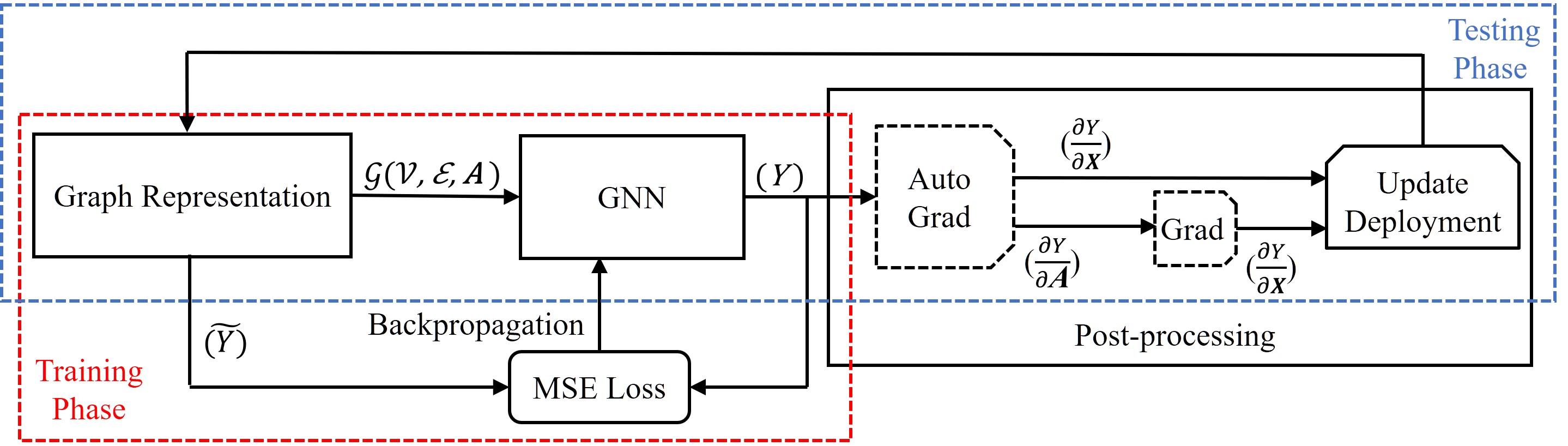}
  \caption{The framework of the proposed MFL method. The red box illustrates the training phase, which learns the max-flow with GNN; the blue box illustrates the testing phase, which updates the relay nodes to maximize the max-flow.}
  \label{GNN_value_framework}
\end{figure*}

After we fit the GNN to the max-flow value, the test phase is conducted as follows. First, the graph representations $\gG(\gV, \gE, \mA)$ of test node locations are fed into the trained GNN. Then, we perform auto-differentiation to obtain the gradients of the output with respect to both the node features $\mX$ and the weighted adjacency matrix $\mA$ 
\begin{equation}
    \frac{\partial g(\mX,\mA)}{\partial x_i^{(m)}} = \bigg\langle\frac{\partial g(\mX,\mA)}{\partial \mX}, \frac{\partial \mX}{\partial x_i^{(m)}}\bigg\rangle +  \bigg\langle\frac{\partial g(\mX,\mA)}{\partial \mA}, \frac{\partial \mA}{\partial x_i^{(m)}}\bigg\rangle,
\end{equation}
where $\ip{\cdot}{\cdot}$ is the Hilbert-Schmidt inner product and
\begin{equation}\label{A_gradient}
    \frac{\partial \mA}{\partial x_i^{(m)}} = \bigg[\frac{\partial A_{p,q}}{\partial x_i^{(m)}}\bigg]_{p,q=1}^{n}.
\end{equation}
Here, $\displaystyle \frac{\partial A_{p,q}}{\partial x_i^{(m)}} $ depends on the network of interest and an example is given in Section \ref{subsec:experiment}.
Finally, the relay nodes are updated along these directions with stepsize $\zeta$ as follows
\begin{equation}
    x_i^{(m)}[\mathfrak{t}+1] \leftarrow x_i^{(m)}[\mathfrak{t}] + \zeta \frac{\partial_i^{(m)} g(\mX[\mathfrak{t}], \mA[\mathfrak{t}])}{\norm{\partial_i^{(m)} g(\mX[\mathfrak{t}], \mA[\mathfrak{t}])}}, ~m = 1, 2,
\end{equation}
where $\partial_i^{(m)} g = {\partial g}/{\partial x_i^{(m)}}$.

\subsection{A Hybrid Method}\label{subsec:hybrid_method}
To further improve the performance of throughput optimization, we also study a hybrid method based on an existing method \cite{he2014dynamic} and the proposed MFL method.
\subsubsection{Existing method}\label{subsec:existing_method}
In spectral graph theory, the Cheeger constant of the non-normalized Laplacian matrix $ \mL $ is an approximation to the max-flow $C$ \cite{sherman2009breaking}. This Cheeger constant can be bounded by an inequality of the second smallest eigenvalue $\lambda_{2}(\mL)$ of the Laplacian matrix, which depends on the locations of the nodes \cite{bhattacharya2010graph}. 
To better capture the max-flow between the source node and the destination node \cite{he2014dynamic}, a weighted Cheeger constant (WCC) is employed in the approximation of $C$. Specifically, let  $\mW$ be the diagonal weight matrix, one defines the weighted Laplacian matrix to be $ \mL_{\mW} = \mW^{-1/2 } \mL \mW^{-1/2 }$.
The relay nodes can be then updated to maximize the second small eigenvalue, $\lambda_{2}(\mL_{\mW})$, of the weighted Laplacian matrix to optimize $C$.
 
\subsubsection{Hybrid method}\label{subsec:hybrid_method_sub}
The hybrid method adopts a greedy algorithm that chooses the one with larger increment of max-flow between the GNN-based MFL method and the WCC method in each step, as  shown in Fig.~\ref{Hybird_method}.
In particular, two different sets of gradients (i.e., the node update direction corresponding to the max-flow increase), Gradients 1 and Gradients 2, are obtained using the GNN-based MFL and WCC methods, respectively.
Two candidate deployments, Deployment 1 and Deployment 2, can be obtained according to these gradients. By comparison, the set of gradients that leads to a higher max-flow will be selected and the corresponding candidate deployment will be taken as the new deployment.
We keep iterating until a pre-specified number of steps is reached.
This hybrid method is expected to synthesize the advantages of both the GNN-based MFL method and the WCC method.

\begin{figure}[!t]
  \centering
  \includegraphics[width=2in]{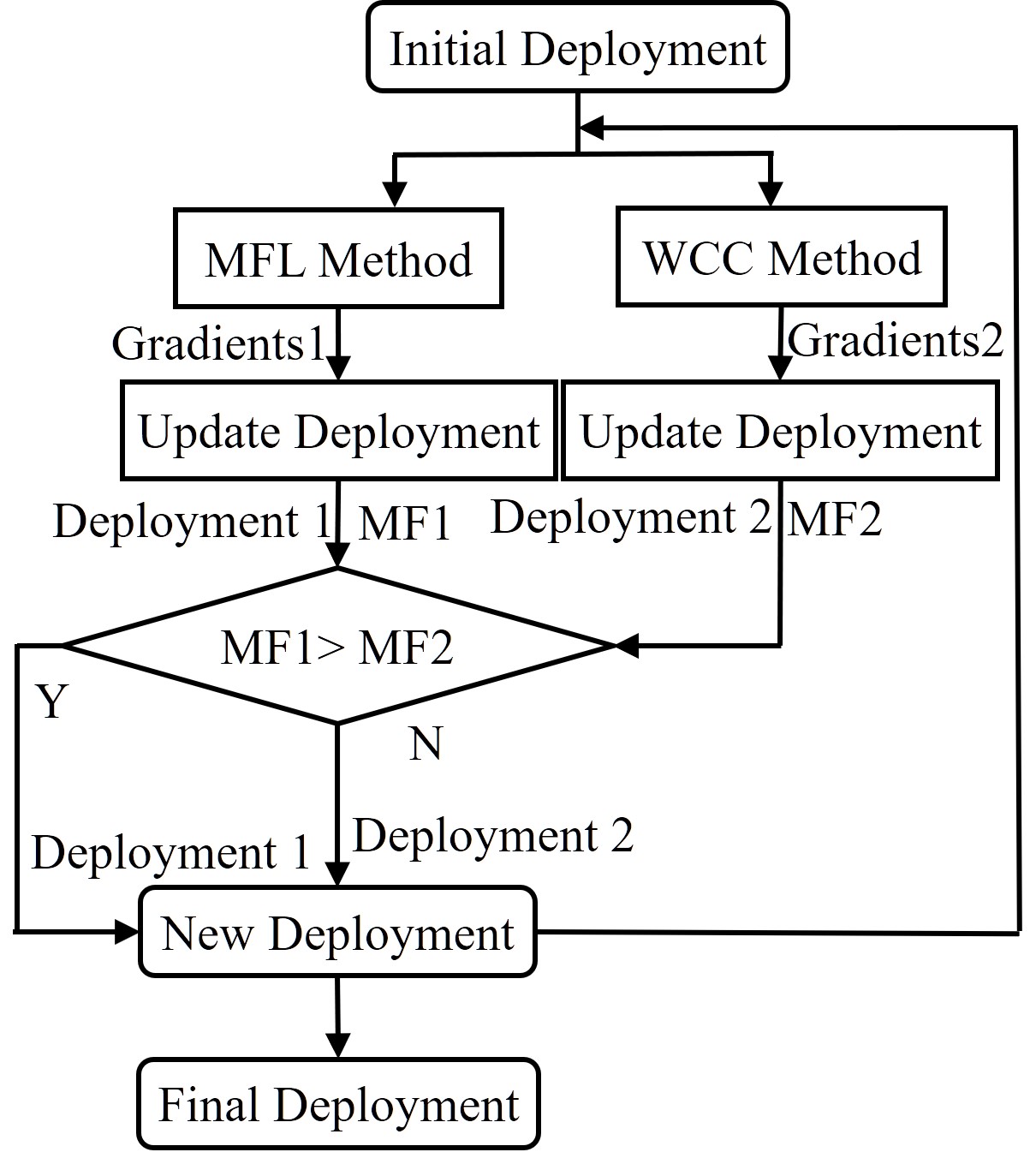}
  \caption{Flow chart of the hybrid method.}
  \label{Hybird_method}
\end{figure}

\subsection{Training Dataset Creation}\label{subsec:dataset_creation}
To train GNN to learn the value of $C$ (or, equivalently, the max-flow), it is vital to create a dataset that well represents max-flow values for different node locations. In particular, it is important that the training dataset contains some node locations corresponding to good max-flow values, so that the GNN will not underestimate the maximum possible value of the max-flow of the network.
To this end, we propose to use a reinforcement learning model based on a GNN implementation of PPO \cite{schulman2017proximal, heess2017emergence} to obtain suitable training samples. 
This algorithm has a strong exploration capability to find the node locations corresponding to some good max-flow values in the sampling space.\footnote{
We remark that PPO could also be used for optimizing the max-flow directly,
but we don’t expect a good performance due to the complexity of the search space (see Section \ref{subsec:compar_RL} for
a comparison of performance). Therefore, the application of PPO is only regarded as a pre-training step which provides us with a good variety of training data.}

In our context, the three key elements of the PPO algorithm are defined as follows:
\begin{itemize}
	\item state ($s$): the graph-structured data including node features $\mX$ and edge features $\mA$;
	\item action ($a$): a collection of two-dimensional vectors representing the update directions of the relay nodes;
	\item reward ($r$): the increment of the max-flow by taking an action at a given state.
\end{itemize}
The agent in the PPO algorithm consists of an actor (i.e. the policy) and a critic (i.e. the value function), both modeled as GNNs. The actor GNN provides a policy for choosing an action for a given state, while the critic GNN evaluates possible cumulative rewards for a state. 
In particular, given a current state as the input, the outputs of the actor GNN are used as means and variances of $ (n-2) \times 2 $ normal distributions $\mathrm{Norm}(\mu_{i,l}, \sigma_{i,l})$, $i = 2, \dots, n-1$, $l=1,2$. The update of the relay node $\vxi_{i}$, $i = 2, \dots, n-1$, is then a vector $(\Delta x_{i,1}, \Delta x_{i,2})$, sampled according to $\Delta x_{i,1} \sim \mathrm{Norm}(\mu_{i,1}, \sigma_{i,1})$ and $\Delta x_{i,2} \sim \mathrm{Norm}(\mu_{i,2}, \sigma_{i,2})$. On the other hand, the critic GNN outputs a deterministic scalar value as the cumulative rewards.

The PPO algorithm runs for $N$ epochs, and each epoch is divided into $T$ timeslots.
Here $T$ is equal to the product of the number of update steps $\mathfrak{T}$ and the number of updates from scratch $\varphi$.
For the $i$-th epoch, $i = 1,2,\dots,N$, we use $\pi _{\theta _{i}}$ to denote the actor GNN and $V_{\phi _{i}}$ to denote the critic GNN, where $\theta_i$ and $\phi_i$ denote the parameters in these two GNNs, respectively.
We also use $s_{t}$, $a_{t}$, and $r_{t}$ to denote the state, action, and reward at the $t$-th timeslot, $t=1,2,\dots,T$, respectively, where for brevity we omit their dependence on $i$.
Moreover, let $\pi _{\theta _{i} } \left ( a_{t} | s_{t}  \right ) $ denote the probability that the policy $\pi _{\theta _{i} }$ executes $a_{t}$ at a given $s_{t}$, and $s'_{t}$ denote the state after executing $a_{t}$ at $s_{t}$. 
With the above notations, we can represent the advantage $A_{t}\left (  \theta_{i} \right )$ as
\begin{equation}
  \label{advantage}
  A_{t}\left (  \theta_{i} \right ) = R_{t} - V_{\phi_{i} }\left ( s_{t}  \right ),
\end{equation}
where $ V_{\phi_{i} }\left ( s_{t}  \right ) $ is the current value function from the critic GNN, and 
\begin{equation}
  \label{targetV}
    R_{t} = r_{t} + \gamma V_{\phi_{i} }\left ( s'_{t}   \right ).
\end{equation}
Here, $\gamma$ is a hyperparameter. Accordingly, the loss function for updating the actor GNN is defined as
\begin{equation}
    \begin{aligned}
     \label{J}
  L_{\rm actor} ( \theta   ) =  \mathbb{E} [  \min  \{ &  \rho_{t} (  \theta  ) A_{t} (  \theta_{i}  ) ,  \\ & \mathrm {clip}   ( \rho_{t} (  \theta  ), 1- \tau  ,1 +  \tau   , A_{t} (  \theta_{i}  ))  \} ],
    \end{aligned}
\end{equation}
where $\mathbb{E}  \left[ \cdot\right ]$ is the empirical average over all samples in the buffer (stored the training data of PPO), $ \tau $ is a hyperparameter, and $\mathrm {clip}(\cdot,\cdot,\cdot,\cdot) $ is a clipping function defined as
\begin{equation}
  \label{eq:clip_function_def}
    \mathrm{clip}(a,b,c,d) = 
    \begin{cases}
    \min \{a,c\} \cdot d   & \text{ if } d>0 \\
    \max \{a,b\} \cdot d  & \text{ otherwise }
    \end{cases},
\end{equation}
and $\rho_t$ is a quotient of probability densities defined as
\begin{equation}
  \label{probability}
\rho_{t}\left (  \theta \right ) = \frac{\pi _{\theta  } \left ( a_{t} | s_{t}  \right ) }{\pi _{\theta _{i} } \left ( a_{t} | s_{t}  \right )}.
\end{equation}
On the other hand, the loss function for updating the critic GNN is defined as

\begin{equation}
  \label{L}
  L_{\rm critic}\left ( \phi  \right ) = \mathbb{E}  \left[ \mathfrak{L} \left ( V_{\phi }\left ( s_{t}  \right ) - R_{t}   \right )\right ],
\end{equation}
where $\mathfrak{L}$ is a piecewise function on $\R$ defined as
\begin{equation}
  \label{loss}
  \mathfrak{L} \left ( x   \right )= \begin{cases}
    0.5 \abs{x}  ^{2}   & \text{ if } \left | x  \right | < 1 \\
    \left |  x  \right |- 0.5  & \text{ otherwise }
    \end{cases}.
\end{equation}

\begin{algorithm}[h]
\caption{PPO \cite{schulman2017proximal, heess2017emergence}}\label{alg:alg1}
\begin{algorithmic}
\STATE
\STATE   \textbf{Initialize:} The parameters $\theta_{0}$ of the policy $ \pi $ (i.e. the actor GNN); the parameters $\phi_{0}$ of the value function $ V $ (i.e. the critic GNN); iteration epochs $N, M, T$.
\STATE $ \textbf{for} $ $ i=1,2,\dots,N $ $\textbf{do}$
\STATE \hspace{0.5cm}$ \textbf{for} $ $ t=1,2,\dots,T $ $\textbf{do}$  
\STATE \hspace{1cm}$ \textbf{Step 1:} $ 
\STATE \hspace{1cm}Based on $ \pi _{\theta _{i} } $, collect $ \left \{ s_{t}, s'_{t} , a_{t}, r_{t}, \pi _{\theta_{i}  } \left ( a_{t} | s_{t}  \right )  \right \}  $ to
\STATE \hspace{1cm}buffer,
\STATE \hspace{1cm}\textbf{if} $t$ mod $\mathfrak{T} \ne 0$ \textbf{then}
\STATE \hspace{1.5cm}$ s_{t+1} \leftarrow s'_{t}$
\STATE \hspace{1cm}\textbf{else}
\STATE \hspace{1.5cm}$ s_{t+1} \leftarrow s_{1}$
\STATE \hspace{1cm}$ \textbf{Step 2:} $ 
\STATE \hspace{1cm}Estimate $A_{t}\left (  \theta_{i} \right )$ from $V_{\phi _{i} } $ according to \eqref{advantage}.
\STATE \hspace{0.5cm}$ \textbf{end for} $ 
\STATE \hspace{0.5cm}$ \textbf{for} $ $ j=1,2,\dots,M $ $\textbf{do}$ 
\STATE \hspace{1cm}$ \textbf{Step 3:} $ 
\STATE \hspace{1cm} $\theta_{i,j} \leftarrow \argmax L_{\rm actor}(\theta) $ \eqref{J}.
\STATE \hspace{1cm}$ \textbf{Step 4:} $ 
\STATE \hspace{1cm} $\phi_{i,j} \leftarrow  \argmin L_{\rm critic}(\phi) $ \eqref{L}.
\STATE \hspace{0.5cm}$ \textbf{end for} $ 
\STATE \hspace{0.5cm}$ \theta_{i+1} \leftarrow \theta_{i,j} , \phi_{i+1} \leftarrow \phi_{i,j}  $
\STATE \hspace{0.5cm}Empty buffer
\STATE $ \textbf{end for} $ 
\end{algorithmic}
\label{alg1}
\end{algorithm}

\begin{figure*}[!t]
  \centering
  \includegraphics[width=.7\textwidth]{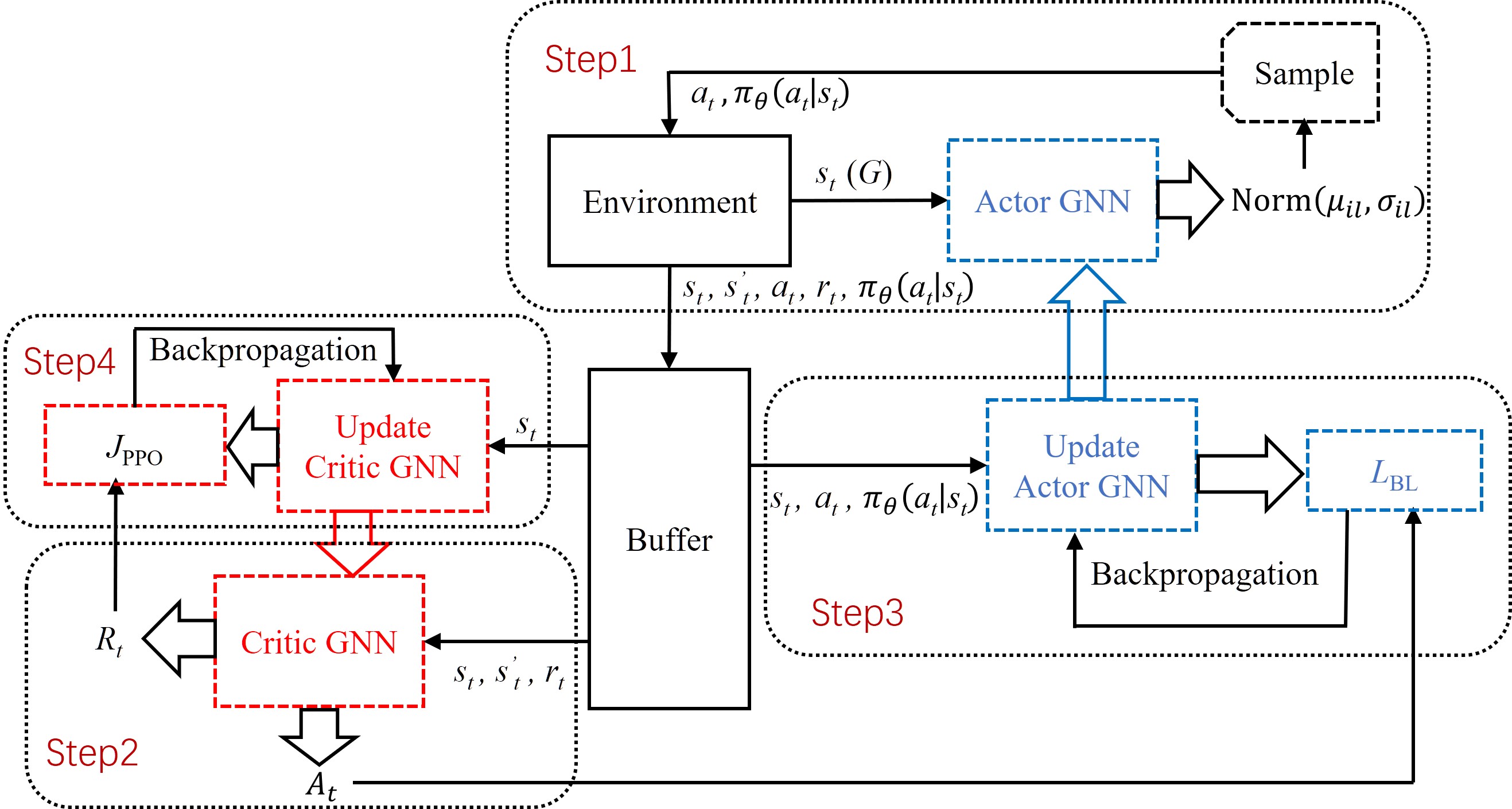}
  \caption{GNN-based PPO framework. 
Step 1: exploration based on the current state $s_{t} $ and the current actor $\pi_\theta$ to obtain $\pi _{\theta _{i} } ( a_{t} | s_{t}  ) $, from which $a_t$ is sampled. Following the update of $a_t$, we obtain $r_{t} $ and $s_{t+1}=s'_{t}$.
These data are stored in the buffer until it is filled.
Step 2: estimate $A_{t}$ according to \eqref{advantage} and $R_{t} $ according to \eqref{targetV} with all data in the buffer based on the current Critic.
Steps 3 and 4: based on the data in the buffer and the estimated $A_{t}$ and $R_{t} $, the actor and critic are updated multiple times according to \eqref{J} and \eqref{L}, respectively.}
  \label{PPO2_framework}
\end{figure*}
For clarity, we summarize the above procedures in Algorithm~\ref{alg:alg1} and Fig.~\ref{PPO2_framework}.
Based on this framework, we explore the state space in a comprehensive way so that a sufficiently rich collection of update trajectories of the relay nodes can be obtained.
Each trajectory $\{(\vxi_1[\mathfrak{t}], \vxi_2[\mathfrak{t}], \dots, \vxi_n[\mathfrak{t}]), \mathfrak{t} = 0,\dots,\mathfrak{T}\}$ consists of the positions of legitimate nodes that are updated $\mathfrak{T}$ times.
To construct the training dataset, we collect data points $\{(\vxi_1[{\mathfrak{t}}{k}], \vxi_2[{\mathfrak{t}}{k}], \dots, \vxi_n[{\mathfrak{t}}{k}]), \mathfrak{t} = 0,\dots,\mathfrak{T}/k\}$ along each trajectory at a fixed time interval $k$.

\section{Experiments}\label{sec:experiment}
\subsection{Experimental setup}\label{subsec:experiment}
We test our methodology in the scenario where, in addition to a wireless network with the source, destination, and relay nodes (which will be referred to as legitimate nodes altogether),
there also exists a jammer. 
Note that similar anti-jamming communication scenarios are commonly considered in the wireless communication literatures \cite{wang2018trajectory, valiulahi2020multi, he2014dynamic, rahmati2021dynamic}.
In the considered scenario, the jammer is assumed to be positioned at a known location $\vxi_{\rm J} = \left( x_{\rm J}^{(1)}, x_{\rm J}^{(2)} \right) \in \R^2 $ and it interferes the data transmission within the wireless network.
The signal-to-interference ratio (SIR) for the transmission from node $i$ to node $j$, ${\rm SIR}_{i,j}$, $i, j \in \{1, \dots, n\}$ can be modelled as
\begin{equation}
\label{SIR}
    {\rm SIR}_{i,j} \overset{\bigtriangleup }{= } \frac{d_{i,j }^{- \alpha } }{\eta  d_{j,{\rm J} }^{- \alpha }+  \sum_{k\in \Xi _{i,j}  }\nu \left ( {d_{j,k} } / {r_{\rm int} }  \right ) },
\end{equation}
where $d_{i,j} = \norm{\vxi_i - \vxi_j}$ is the distance between node $i$ and node $j$, $d_{j, J} = \norm{\vxi_j - \vxi_{\rm J}}$ is the distance between node $j$ and the jammer, 
$\eta$ is the power ratio of the jammer to the legitimate nodes. Note that, a path-loss model \cite{goldenberg2004towards} for the wireless channel is assumed in \eqref{SIR} and $\alpha$ is the path loss factor, and $r_{\rm int}$ is a design parameter.
Also, $\Xi_{i,j} = \{1, \dots, n\} \backslash \{i,j\}$, where ``$\backslash$'' denotes set difference, is the set of indices of legitimate nodes excluding $i$ and $j$.
To prevent violation of physical constraints, $\nu$ is defined by
\begin{equation}
  \label{smoothed_step_function}
  \nu \left ( z \right )   \overset{\triangle}{=} \varrho \frac{ \exp(- \kappa  z-\log  z_{0} ) }{1+  \exp(- \kappa  z-\log z_{0} )},
\end{equation}
where $\varrho$ and $ \kappa$ are all design parameters, and $ z_{0} $ can be an arbitrarily  small positive number.
Based on the above model, the (two-way average) communication rate $f_{\rm cap}$ in \eqref{aij} between two nodes $i$ and $j$ is defined as
\begin{equation}
  \label{f_function}
  f_{\rm cap} (\vxi_i, \vxi_j) \overset{\triangle}{=} \frac{B} { \left [   \ln{\left ( 1+  {\rm SIR}_{i,j} \right )  }  \right ] ^{- 1 }+ \left [   \ln{\left ( 1+  {\rm SIR}_{j,i} \right )  }  \right ] ^{- 1 }  },
\end{equation}
where it is assumed that orthogonal spectrum is allocated to each link of the network and $ B $ denotes the corresponding bandwidth per link.
Therefore, the corresponding derivatives $\displaystyle \frac{\partial A_{p,q}}{\partial x_i^{(m)}} $ (cf. \eqref{A_gradient}) for this considered scenario can be found by the chain rule based on \eqref{SIR}--\eqref{f_function}.

In our experiments, we consider a wireless network with $n=6$ legitimate nodes in a $600\times 600 ~m^{2} $ area.  
We represent this area by a Cartesian coordinate system, where the range of $x$-coordinate is $\left ( -6, 6 \right)$ and the range of $y$-coordinate is $\left ( -6, 6 \right)$, and the unit length in both coordinates is set to 50 meters.
The coordinates of the source and the destination nodes are fixed to $\vxi_1 = (-4.5, 0) $ and $\vxi_6 = (4.5, 0)$, respectively. 
In each deployment, the relay nodes are initialized with the following coordinates: 
$\vxi_2 = (-2.7, 0)$, $\vxi_3 = (-0.9, 0)$, $\vxi_4 = (0.9, 0)$, $\vxi_5 = (2.7, 0)$.
The coordinates of the jammer node are randomly chosen with the following constraints: $ x_{\rm J}^{(1)} \in \left ( -6, 6 \right), x_{\rm J}^{(2)} \in \left ( -6, 6 \right), d_{1,{\rm J}}>3, $ and $ d_{6, \rm{J}}>3$, where the last two constraints represent ``guard zones'' that prevent the jammer from getting too close to the source or destination.

To generate the training set, we initialize a total of $2,000$ deployments. Every deployment is updated for $\mathfrak{T} = 400$ steps with stepsize $\zeta=0.02$ following the PPO strategy described in Section \ref{subsec:dataset_creation}. In the GNN-based PPO, we choose $\gamma=0.9$, $\tau=0.2$, $\varphi=5$, $T=400\times 5=2,000$, $M=10 $.
We collect a data point once every $5$ steps from each deployment. For convenience, we name this dataset the RLGP (short for Reinforcement Learning GNN-PPO) dataset. For efficiency, the actor GNN and the critic GNN stop training when the reward converges or the number of epochs reaches the maximum. 
On the other hand, the test set contains $500$ initial deployments which have different jammer locations.

We used the PyTorch Geometric library \cite{fey2019fast} to implement all the GNN-related tasks.

\subsection{Baseline Methods}\label{subsec:baseline}
We compare MFL with the following baseline methods. 
\begin{itemize}
    \item \emph{Gradient Learning (GL)}: This baseline is the na\"{i}ve approach introduced in Section \ref{subsec:naive_approach}. The GNNs in both MFL and GL are trained using the RLGP dataset described in Section \ref{subsec:experiment}. The label for MFL is the max-flow value, while the labels for GL are the directions of the relay nodes traveled in the training procedure.
    \item \emph{Reinforcement Learning (RL)}: This baseline is a direct application of PPO to the test data. Specifically, we train an agent to explore the test data and provide the update directions for the relay nodes. The training procedure still follows the framework illustrated in Fig.~\ref{PPO2_framework}.
    We obtain the optimal trajectory of each test deployment and use it to compare with MFL. To ensure good results, the number of timeslots, $T$, is chosen to be $40,000$, and the number of epochs, $N$, is $10,000$.
    \item \emph{Weighted Cheeger Constant (WCC)}: This baseline is introduced in Section \ref{subsec:existing_method}. As in the CC method, we use the gradients of $\lambda_{2}(\mL_{\mW})$ to adjust the locations of relay nodes.
    In our experiments, the weight matrix is $\mW= {\rm diag}\{ w_{1},\dots , w_{n}\}$, where $w_{i} = 3n$ if $i = 1, n$, and $w_i = 1$ if $i \neq 1, n$.
\end{itemize}

\subsection{GNN Architectures}\label{subsec:gnn_architecture}
In our experiments, we use several different GNNs. We describe their structures as follows.

\subsubsection{GNN in GL}\label{subsec:gnn_GL}
In GL, we use a GNN that consists of a cascading of three GraphConv layers, each followed by the GELU activation function.
The outputs of the last GraphConv layer are then pushed forward through two linear layers, with another GELU activation in between.
This architecture is illustrated in Fig.~\ref{GNN_gradient_model}.

\begin{figure}[!t]
  \centering
  \includegraphics[width=1.75in]{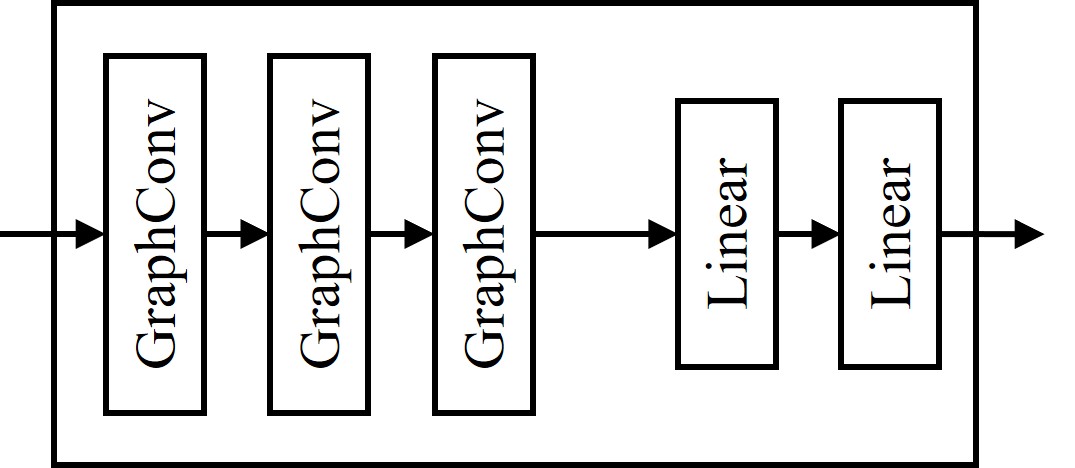}
  \caption{The GNN model for gradient learning is a cascading of three convolutional layers and two linear layers.}
  \label{GNN_gradient_model}
\end{figure}

\subsubsection{GNN in MFL}\label{subsec:gnn_ML}
Since MFL is a graph-level task, we need to construct a  permutation-invariant GNN through a global pooling layer. To this end, we take a global sum of the features of all nodes which returns a graph-level output after all the GraphConv layers. We use again the GELU activation functions between layers. We illustrate this GNN in Fig.~\ref{GNN_value_model}.
\begin{figure}[!t]
  \centering
  \includegraphics[width=2.5in]{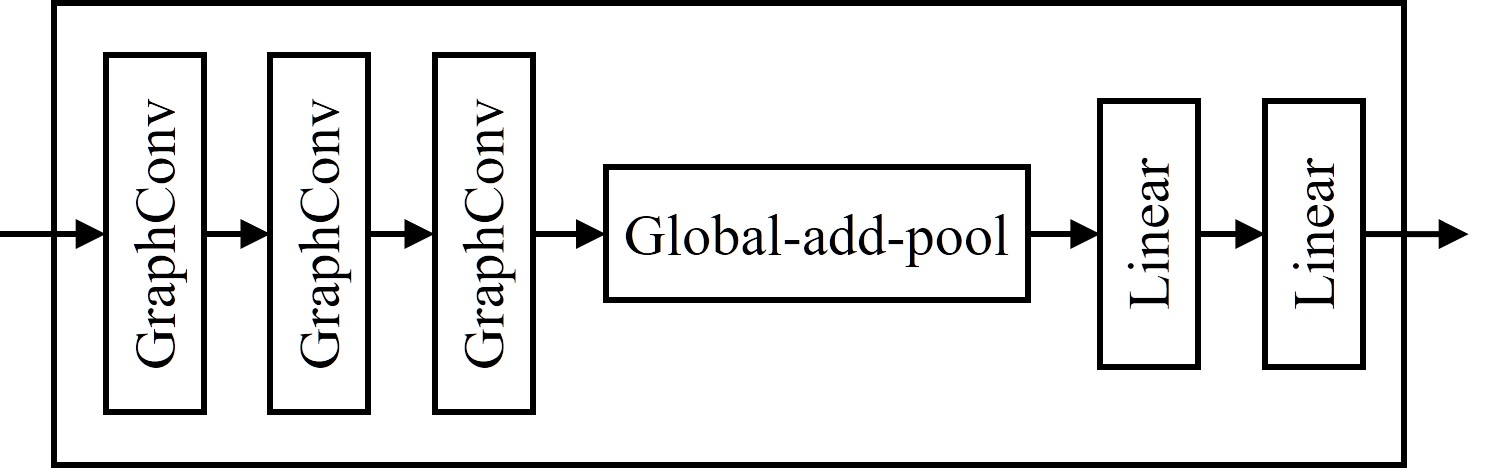}
  \caption{The GNN model for learning the max-flow. We show the architecture used in our experiments, including a cascading of three convolutional layers, one global pooling layer, and two linear layers.}
  \label{GNN_value_model}
\end{figure}

\subsubsection{GNNs in GNN-based PPO}\label{subsec:gnn_ppo}
The structures of the actor GNN and the critic GNN are shown in  Fig.~\ref{gnn_model_ppo}.
The actor GNN consists of a cascading of the following: two GraphConv layers, each with a GELU activation function; global sort pooling; and two different fully-connected linear layers, whose outputs are activated by $\tanh$ and ${\rm softplus} = \log(1+\exp(\cdot))$ respectively. 
The critic GNN has the same graph convolutional layers and activation function, but we use a global addition pooling layer followed by a fully-connected linear layer.

\begin{figure}[!t]
  \centering
  \subfloat[]{\includegraphics[width=2.5in]{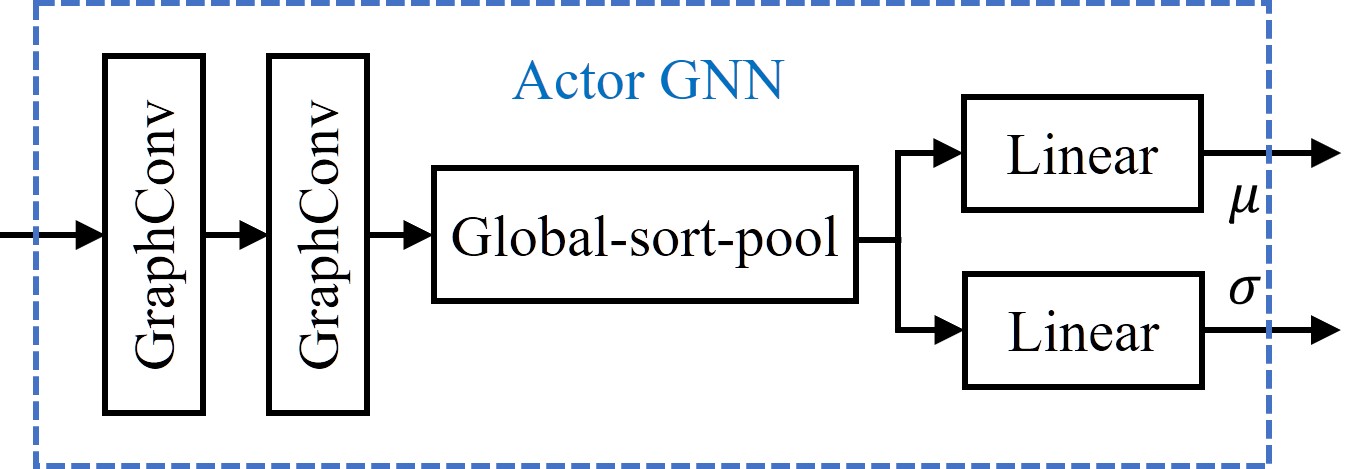}%
  \label{PPO2_Actor_framework}}
  \hfil
  \subfloat[]{\includegraphics[width=2in]{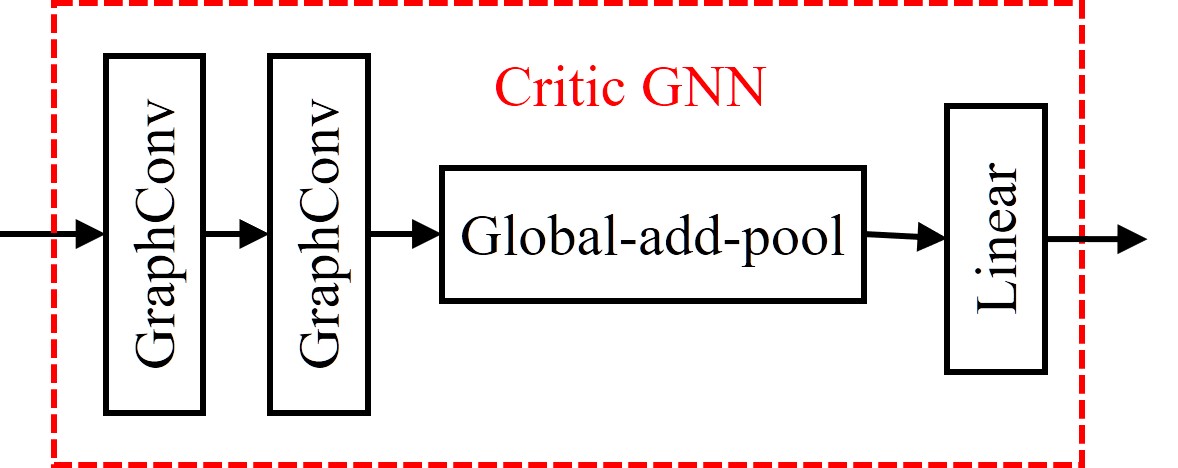}%
  \label{PPO2_Critic_framework}}
  \caption{The GNN model for the PPO framework. (a) The actor GNN model. (b) The critic GNN model.}
  \label{gnn_model_ppo}
\end{figure}

\begin{table}[!t]
\caption{Hyperparameters of GNNs}
\label{tab:hyperparameters_GNN}
\centering
\renewcommand\arraystretch{1.5}
\begin{tabular}{|c||c||c||c||c|}
\hline
              & LG                         & MFL                        & Critic                     & Actor    \\
\hline
GraphConv1    & $(3,32)$                     & $(3,32)$                                   & (3,32)                     & $(3,32)$   \\
\hline
GraphConv2    & $(32,32)$                   & $(32,32)$                                  & $(32,32)$                    & $(32,32)$  \\
\hline
GraphConv3    & (32,32)                    & $(32,32)$                                  & N/A                      & N/A         \\
\hline
Linear1        & $(32,32)$                    & $(32,32)$                                  & $(128,8) $                   & $(32,1) $  \\
\hline
Linear2        & $(32,2)$                     & $(32,1)$                     & $(128,8) $                   & N/A                        \\
\hline
Learning rate & 0.0002 & 0.0002 & 0.0001 & 0.0004                   \\
\hline
Bitch size    & 100    & 100    & 100  & 100                          \\
\hline
Number of epoch         & 5,000   & 8,000    & \multicolumn{2}{c!{\vrule}}{3,000}  \\                     
\hline
\end{tabular}
\end{table}

The hyperparameters of all the GNNs, including both those in the main experiments and those in the dataset creation, are presented in Table~\ref{tab:hyperparameters_GNN}, where the numbers in $(\cdot,\cdot)$ denote the number of input channels and the number of output channels, respectively.

\subsection{Comparison Results}\label{subsec:compar_methods}
We use the same test set described in Section \ref{subsec:experiment} to evaluate the performance of all baseline methods mentioned above.
In Fig.~\ref{update}\subref{map_update}, we illustrate an example showing the trace of updates during the testing process, where the relay nodes are updated for $\mathfrak{T} = 400$ steps.
The corresponding change of the max-flow value of this network is shown in Fig.~\ref{update}\subref{maxflow_update}.
The performance of the methods is measured by the final max-flow of the network after the relay node location update.

\begin{figure}[!t]
  \centering
  \subfloat[]{\includegraphics[width=1.75in]{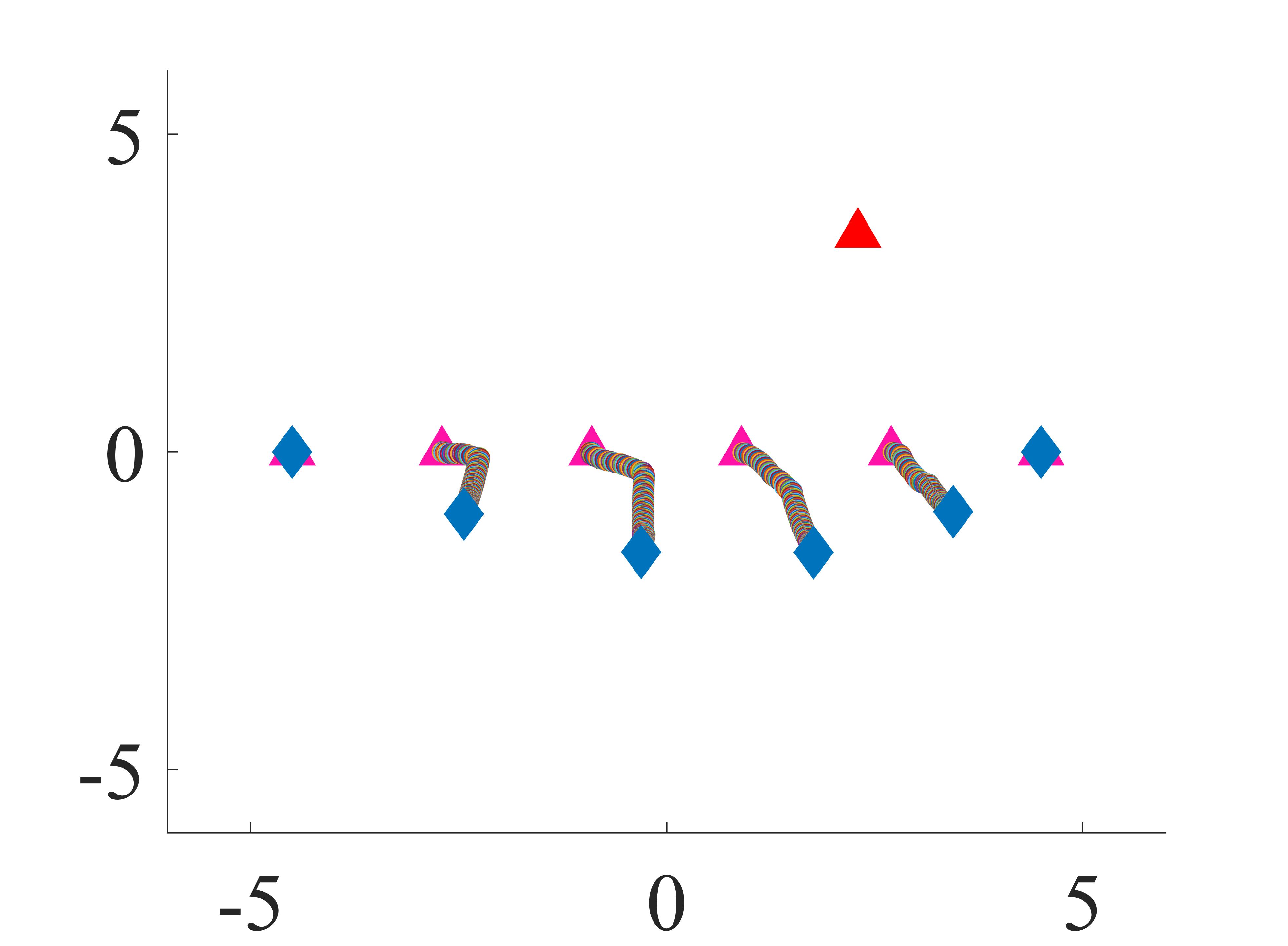}%
  \label{map_update}}
  \subfloat[]{\includegraphics[width=1.75in]{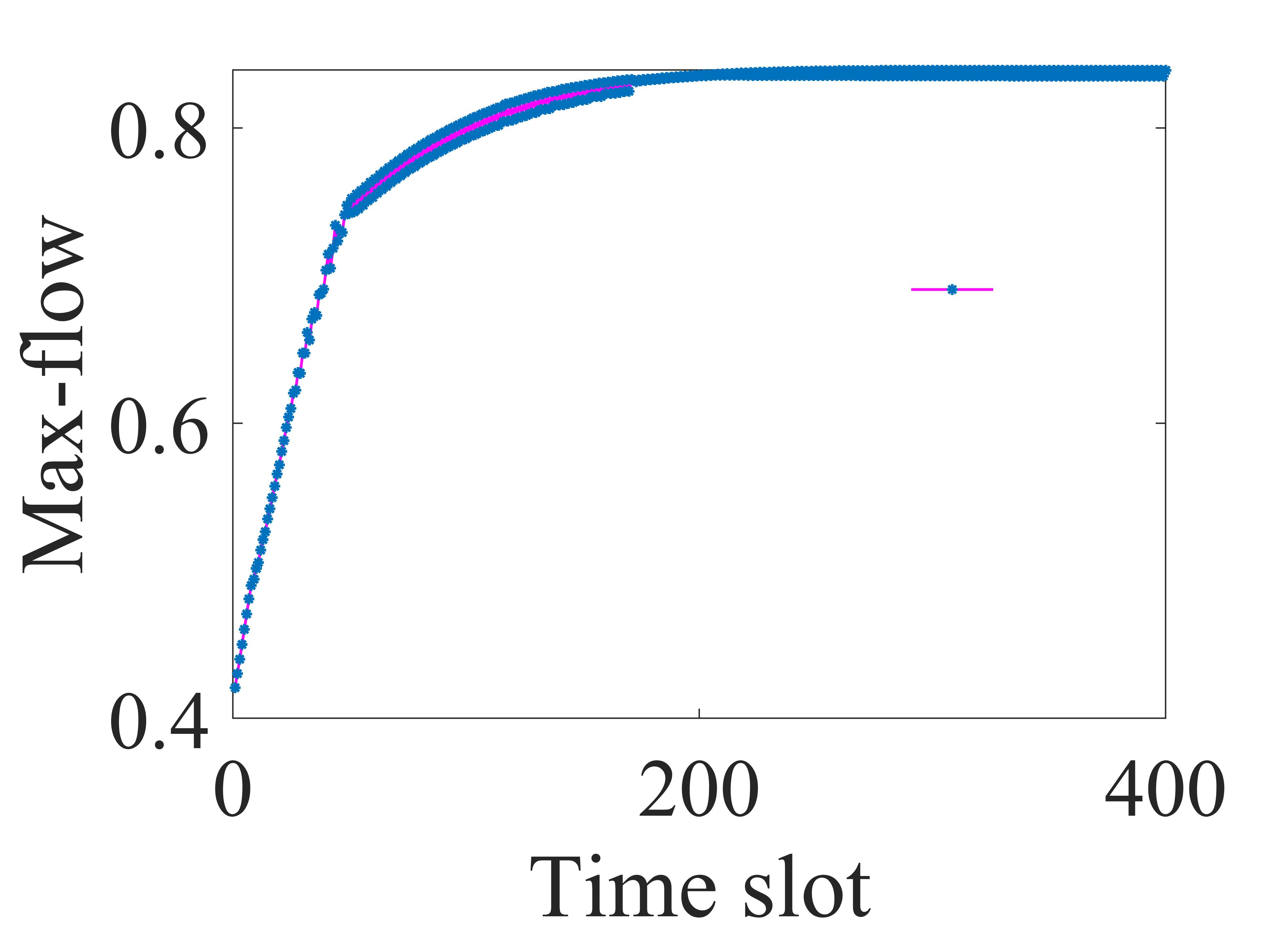}%
  \label{maxflow_update}}
  \caption{An example of the update of the nodes.  (a) The red triangle represents the jammer node, the pink triangles represent the legal nodes of the initial deployment, the blue rhombi represent the legal nodes of the updated final deployment, and the colored circles represent the trace of the relay nodes during the update for $400$ steps. (b) Change of max-flow with nodes update.}
  \label{update}
\end{figure}

\subsubsection{Comparison with GL}\label{subsec:compar_GL}
We first compare the performance between MFL and GL. 
In Fig.~\ref{result_compar_GL}\subref{ML494_GL6_his}, we summarize the relative difference of the final max-flow between the two methods in a histogram. Here, a positive number in the horizontal axis (``$>0$'') indicates that the MFL method is superior and a negative number (``$<0$'') indicates the contrary.
There are $494$ deployments out of $500$ in which MFL outperforms GL, and in $420$ of them, the final max-flow value from MFL is at least $20\%$ higher than GL.
There are also a few deployments where MFL is significantly superior. We further illustrate the locations of the jammer and imply the better performing method in Fig.~\ref{result_compar_GL}\subref{ML_GL_dis}. Clearly, only when the jammer node is located near the origin may GL outperform MFL.
\begin{figure}[!t]
  \centering
  \subfloat[]{\includegraphics[width=1.75in]{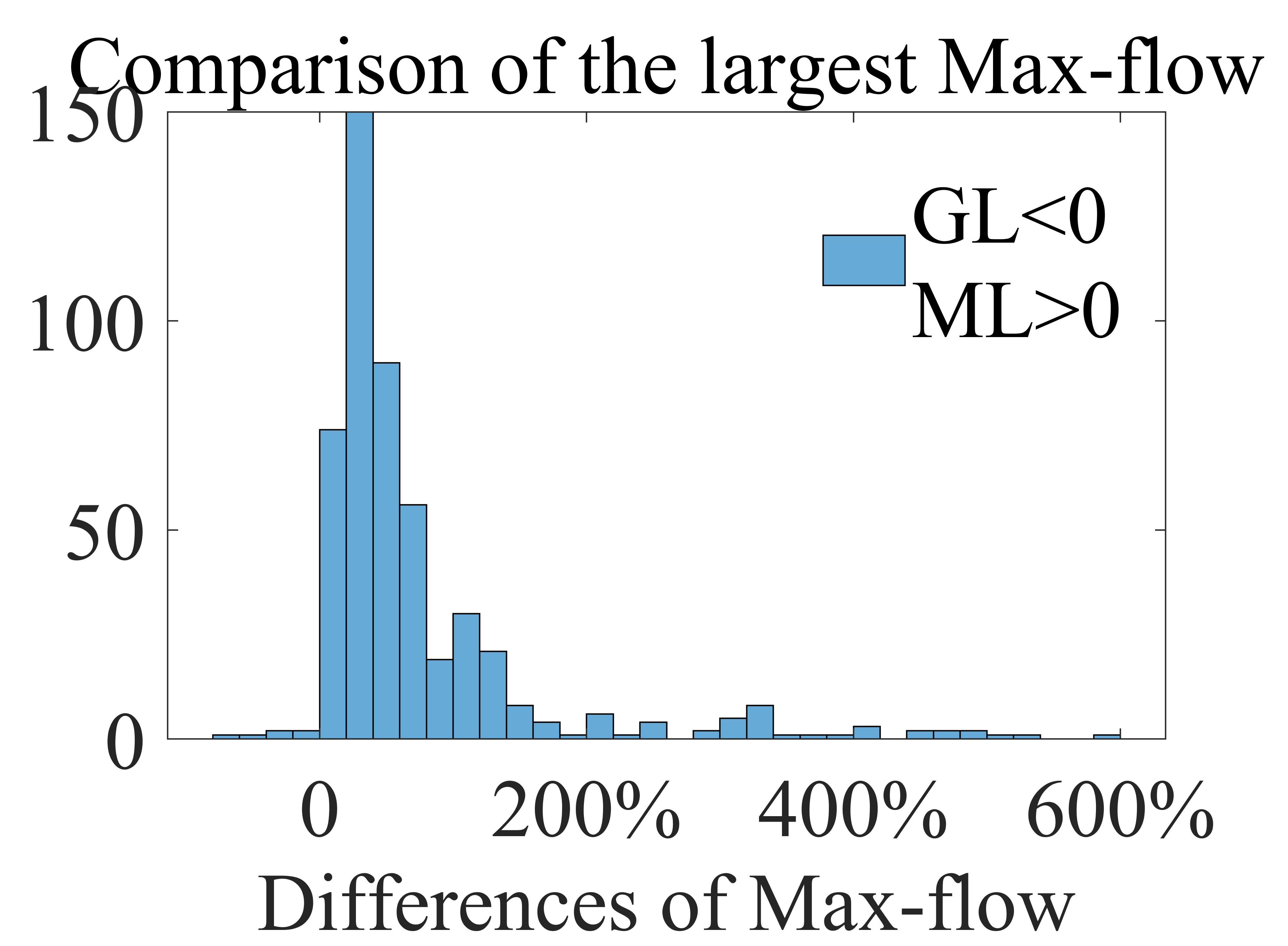}%
  \label{ML494_GL6_his}}
  \subfloat[]{\includegraphics[width=1.75in]{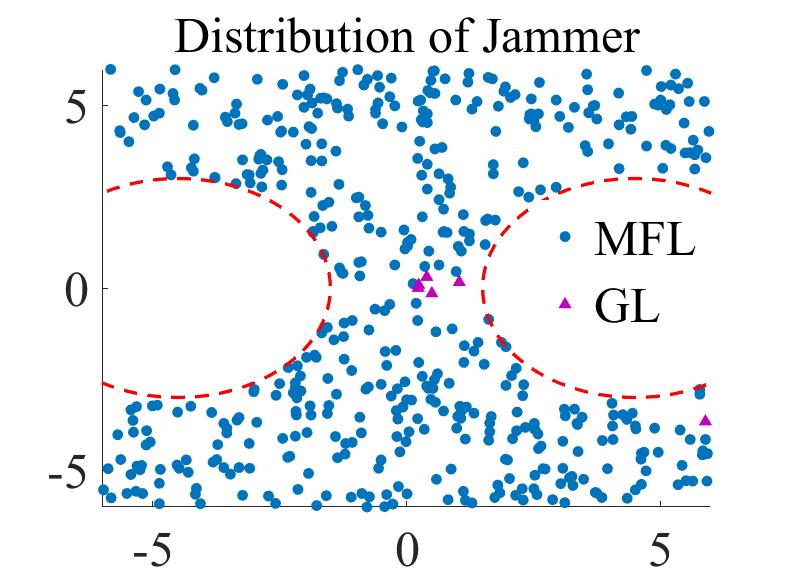}%
  \label{ML_GL_dis}}
  \caption{Comparison with the GL method. (a) Histogram of the max-flow difference. (b) Distribution of jammer nodes: the blue circles indicate the locations of jammer nodes for which the MFL method excels, the pink triangles indicate the locations of jammer nodes for which the GL method excels, and the red circular dashed line indicates the boundary of the guard zones.}
  \label{result_compar_GL}
\end{figure}

\subsubsection{Comparison with RL}\label{subsec:compar_RL}
We compare MFL and RL and show the results in Fig.~\ref{result_compar_RL}.
As shown in Fig.~\ref{result_compar_RL}\subref{ML449_RL51_his}, in $449$ (more than $89\%$) out of $500$ deployments we observe a superior performance from MFL. 
Fig.~\ref{result_compar_RL}\subref{ML_RL_dis} implies that the jammer nodes for which MFL performs worse are mainly located around the edge of the guard zones. 
\begin{figure}[!t]
  \centering
  \subfloat[]{\includegraphics[width=1.75in]{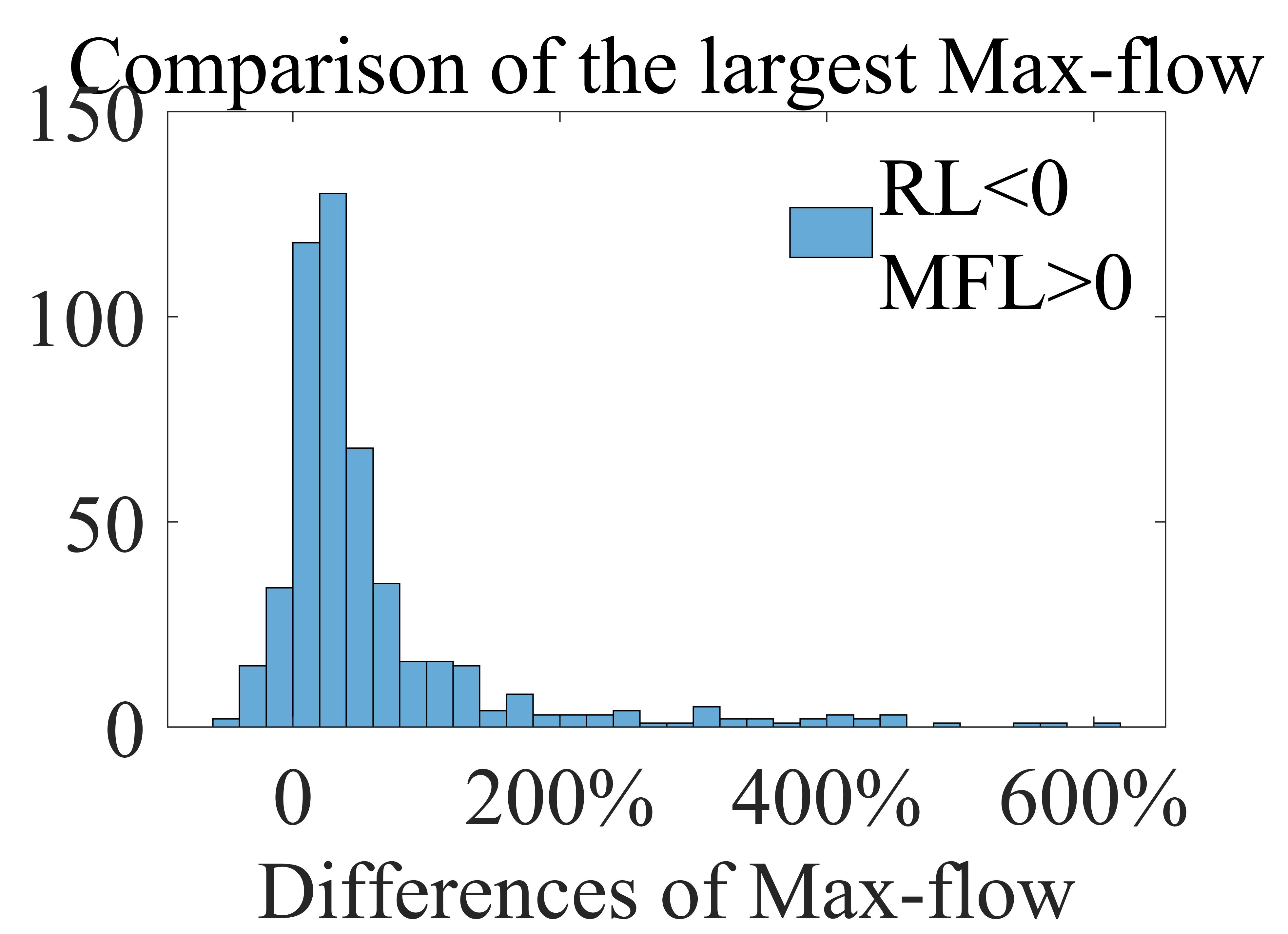}%
  \label{ML449_RL51_his}}
  \subfloat[]{\includegraphics[width=1.75in]{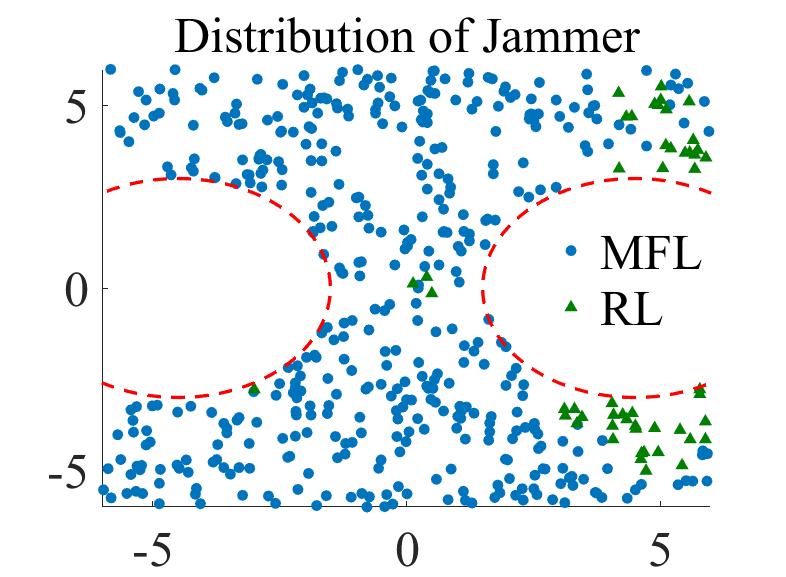}%
  \label{ML_RL_dis}}
  \caption{Comparison with the RL method. (a) Histogram of the max-flow difference. (b) Distribution of jammer nodes: the blue circles indicate the locations of jammer nodes for which the MFL method excels, the green triangles indicate the locations of jammer nodes for which the RF method excels, and the red circular dashed line indicates the boundary of the guard zones.}
  \label{result_compar_RL}
\end{figure}

\subsubsection{Comparison with WCC}\label{subsec:compar_WCC}
Lastly, we compare the performance of MFL and WCC and show the result in Fig.~\ref{result_compar_WCC}.
As shown in Fig.~\ref{result_compar_WCC}\subref{ML377_WCC123_his}, MFL outperforms WWC in $377$ deployments (more than $75\%$) out of $500$. 
From Fig.~\ref{result_compar_WCC}\subref{ML_WCC_dis}, we can observe again that MFL is less effective than WCC mainly at the edge of the guard zones.
\begin{figure}[!t]
  \centering
  \subfloat[]{\includegraphics[width=1.75in]{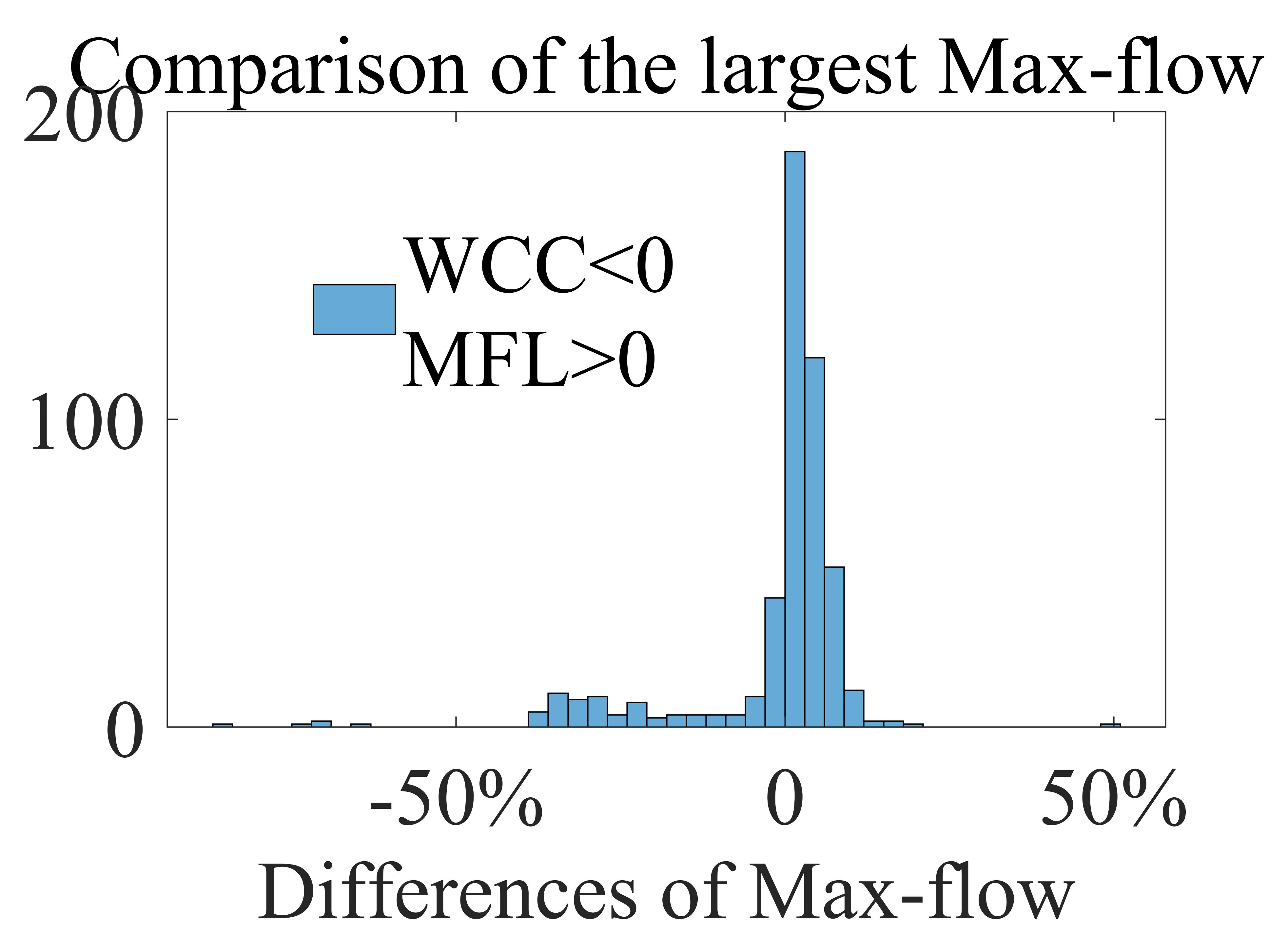}%
  \label{ML377_WCC123_his}}
  \subfloat[]{\includegraphics[width=1.75in]{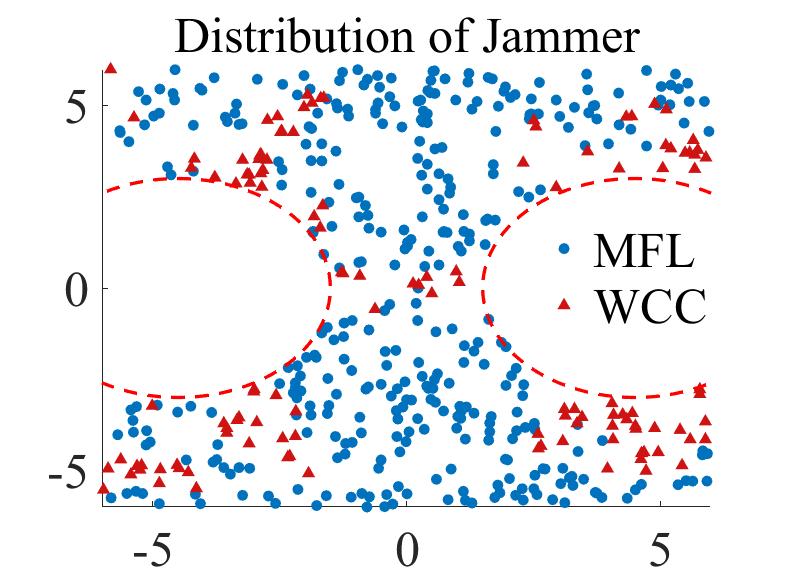}%
  \label{ML_WCC_dis}}
  \caption{Comparison with the WCC method. (a) Histogram of the max-flow difference. (b) Distribution of jammer nodes: the blue circles indicate the locations of jammer nodes for which the MFL method excels, the red triangles indicate the locations of jammer nodes for which the WCC method excels, and the red circular dashed line indicates the boundary of the guard zones.}
  \label{result_compar_WCC}
\end{figure}

\subsubsection{Discussion}\label{subsec:discussion}
It is clear from the results that only for some special locations of the jammer node near the origin or the edge of the guard zones, MFL is not as effective as baseline methods. We believe the reason is that the jammer node in these deployments is too close to the relay nodes, which causes the SIR between the surrounding relay nodes to be close to zero. In such cases, a small change of locations of the relay nodes may lead to a dramatic change in the max-flow. Our GNN is smooth and may not be capable of representing such drastic changes.

\subsection{Results of Hybrid Method}\label{subsec:results_hybrid_method}
We present the experimental results of the hybrid method in Fig.~\ref{result_compar_hybird}. 
\begin{figure}[!t]
  \centering
  \subfloat[]{\includegraphics[width=1.75in]{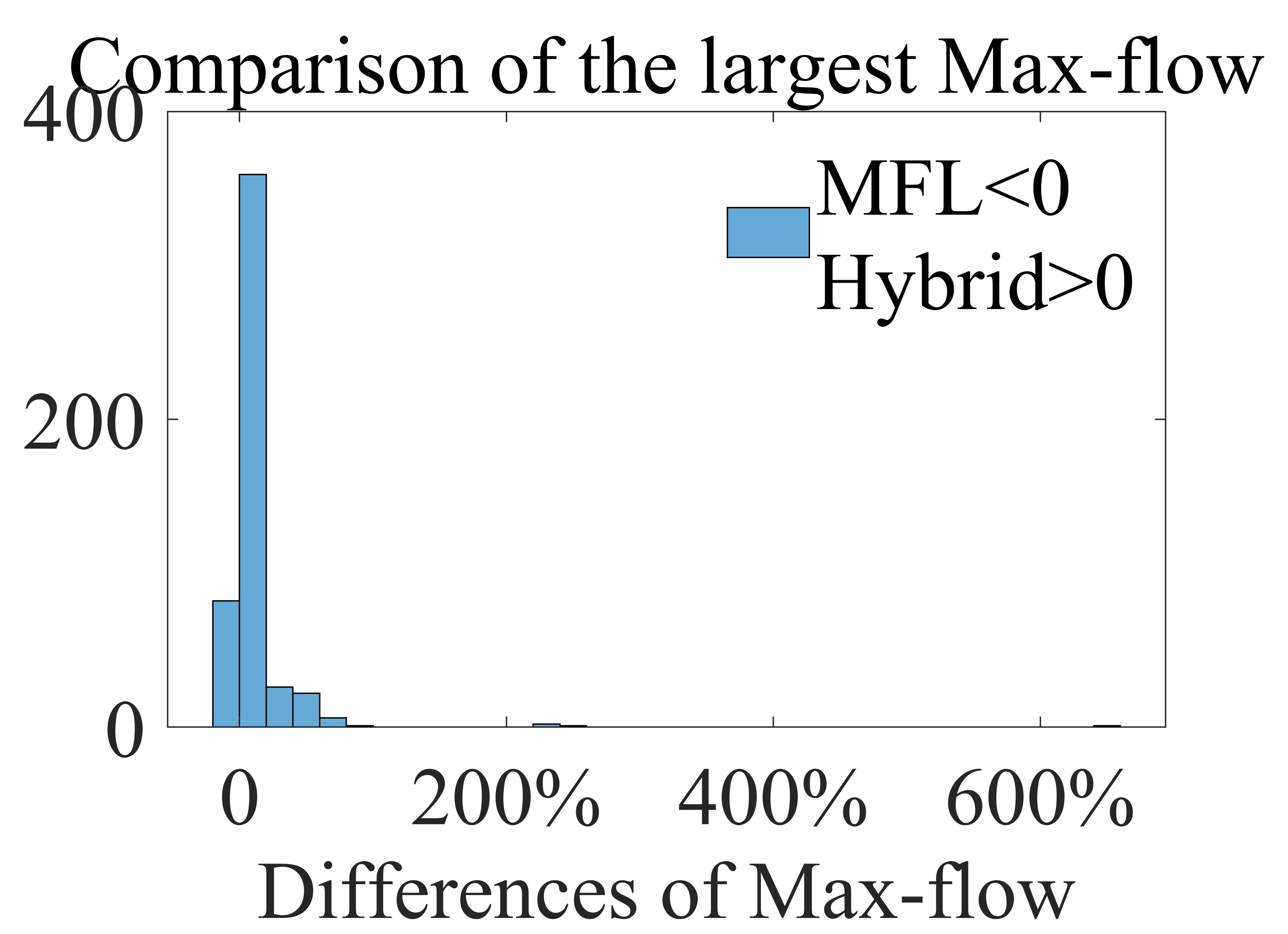}%
  \label{Hybird418_ML82_his}}
  \subfloat[]{\includegraphics[width=1.75in]{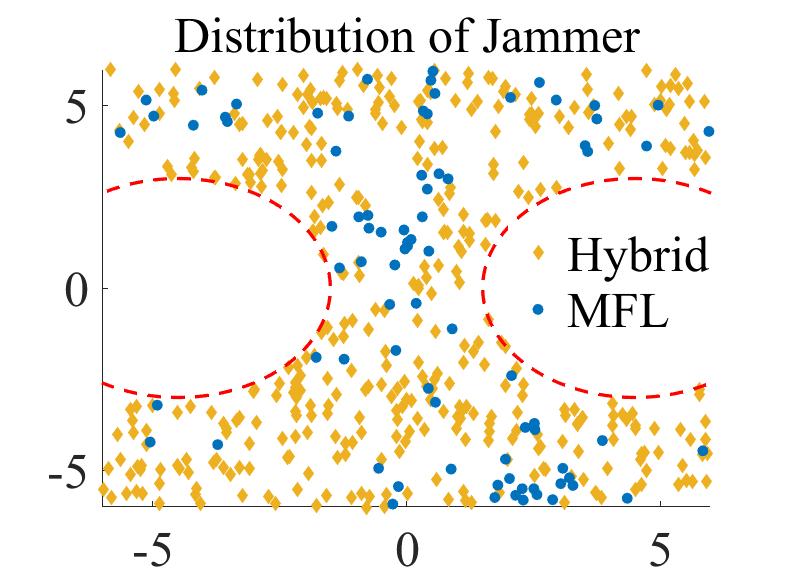}%
  \label{Hybird_ML_dis}}
  \hfil
  \subfloat[]{\includegraphics[width=1.75in]{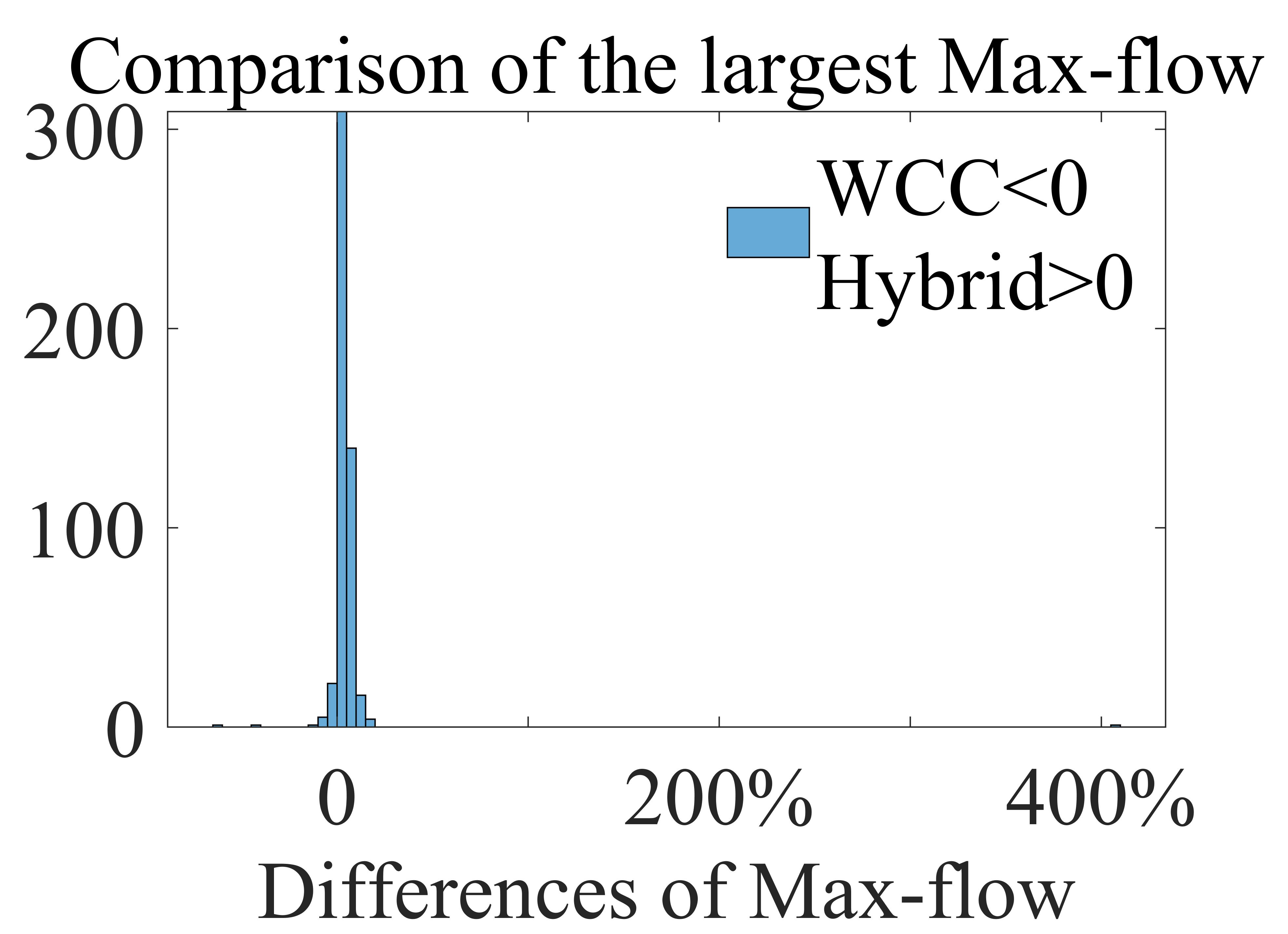}%
  \label{Hybird470_WCC30_his}}
  \subfloat[]{\includegraphics[width=1.75in]{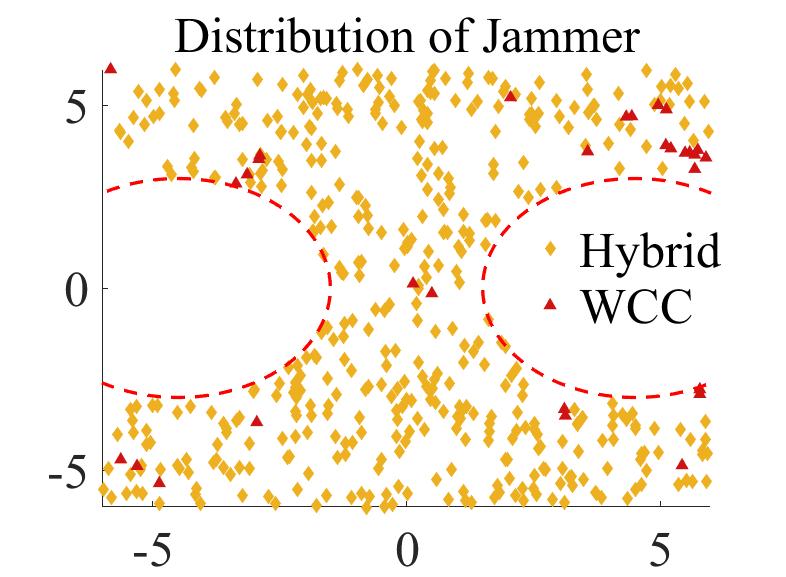}%
  \label{Hybird_WCC_dis}}
  \caption{Comparison with the MFL method and the WCC method. (a) Histogram of the max-flow difference between the hybrid method and the MFL method. (b) Distribution of jammer nodes, the blue circles indicate the locations of jammer nodes for which the MFL method excels, the yellow diamonds indicate the locations of jammer nodes for which the hybrid method excels, and the red circular dashed line indicates the boundary of the guard zones. (c) Histogram of the max-flow difference between the hybrid method and the WCC method. (d) Distribution of jammer nodes, the red triangles indicate the locations of jammer nodes for which the MFL method excels, the yellow diamonds indicate the locations of jammer nodes for which the hybrid method excels, and the red circular dashed line indicates the boundary of the guard zones.}
  \label{result_compar_hybird}
\end{figure}
The performance of the hybrid method excels both individual approaches.
Compared with the MFL method, the hybrid method is superior in $ 418 $ test deployments.
Fig.~\ref{result_compar_hybird}\subref{Hybird418_ML82_his} reveals that although most of the data only see an improvement of $20\%$ or less, there were still some deployments with an improvement of $50\%$ or more, with the largest reaching $641\%$.
In Fig.~\ref{result_compar_hybird}\subref{Hybird_ML_dis}, we observe that some deployments for which the MFL method does not perform well have an improved result when we use the hybrid method.
On the other hand, as observed in  Fig.~\ref{result_compar_hybird}\subref{Hybird470_WCC30_his}, $94\%$ of the test data enjoy improved results when the hybrid method replaces the WCC method.
Although the hybrid method performs well, Fig.~\ref{result_compar_hybird}\subref{Hybird_WCC_dis} shows clearly that there are still specific regions of the jammer node where the hybrid method falls behind the WCC method. We believe the hybrid method does not completely resolve the problem at the edge of guard zones, from which the MFL method suffers.

To summarize the results more clearly, we use the WCC method as a benchmark to quantitatively compare the average performance in the test set. We show the average difference of the max-flow (avg. diff. mf.) and the average relative difference (avg. rel. diff. mf.) in Table~\ref{table:diff_maxflow}.
In addition, to eliminate the influence of outliers, we calculate the truncated-averaged difference (c-avg. diff. mf.) and the truncated-averaged relative difference (c-avg. rel. diff. mf.), excluding the highest and lowest $10\%$ values. In all these metrics, the hybrid method has the best overall performance.
\newcommand{\tabincell}[2]{\begin{tabular}{@{}#1@{}}#2\end{tabular}}

\begin{table}[!t]
  \caption{Comparison of the average final max-flow value using WCC as baseline}
  \label{table:diff_maxflow}
  \centering
  \renewcommand\arraystretch{1.5}
  \begin{tabular}{|c||c|c|c|c|}
    \hline
        ~ & \tabincell{c}{avg.\\ diff. mf. } &  \tabincell{c}{avg. rel.\\ diff. mf. }  & \tabincell{c}{c-avg.\\ diff. mf. }  &  \tabincell{c}{c-avg. rel.\\ diff. mf. }                                          \\ 
    \hline
    WCC  & 0   & 0 & 0  & 0                                          \\ 
    \hline
    GL   & -0.3099      & -35.47\%   & -0.3016  & -33.79\%                         \\ 
    \hline
    RL &  -0.2196   & -27.82\%  & -0.2188   & -25.75\%                                                     \\ 
    \hline
 MFL   & -0.0128  & -1.43\%  & 0.0116 & 1.32\% \\
\hline
 Hybrid  & \bf{0.0314} & \bf{4.21\%} & \bf{0.0301} & \bf{3.49\%} \\
\hline 
\end{tabular}
\end{table}

\subsection{Ablation Studies}\label{subsec:ablation_studies}
We perform ablation studies to show two important ideas implemented in our experiments. The first is the use of reinforcement learning for generating the dataset; the second is the use of expressive GNN layers throughout our experiments.

\subsubsection{Comparison of different datasets}\label{subsec:compar_diff_data}
In addition to the RLGP dataset in our main experiments, we create two additional datasets, namely the random walk (which we abbreviate as ``RW'') dataset and the weighted Cheeger constant (which we abbreviate as ``WCC'') dataset.
The setups of generating two additional datasets are the same as RLGP, except for the strategy for moving the relay nodes.
Specifically, the RW dataset takes a direction uniformly on the unit circle, while the WCC dataset follows the WCC method when choosing the direction. Both datasets follow the same stepsize $\zeta$ when moving the relay nodes.
We use the RW dataset, WCC dataset, and RLCP dataset to train GNN respectively and implement the MFL method with these GNNs. The comparison of their test results is shown in Fig.~\ref{result_compar_diff_data}.
In Fig.~\ref{result_compar_diff_data}\subref{RLGP472_RW28_his}, the histogram reveals that the results using the RLGP dataset are remarkably superior to the RW dataset. $472$ of the results of the RLGP dataset are superior to the results for the RW dataset, and $211$ of them have a final max-flow larger than $40\%$ relative to the RW, with the largest even reaching $713\%$.
From Fig.~\ref{result_compar_diff_data}\subref{RLGP390_WCC110_his}, it is clear that the results using the RLGP dataset excel the RW dataset. Indeed, we observe that in $390$ deployments, which is close to $80\%$ of the total, the final max-flow values are higher when the GNN is trained using the RLGP dataset.
We conclude that using the RLGP dataset makes a significant contribution to the excellent performance of our proposed methodology.
\begin{figure}[!t]
  \centering
  \subfloat[]{\includegraphics[width=1.75in]{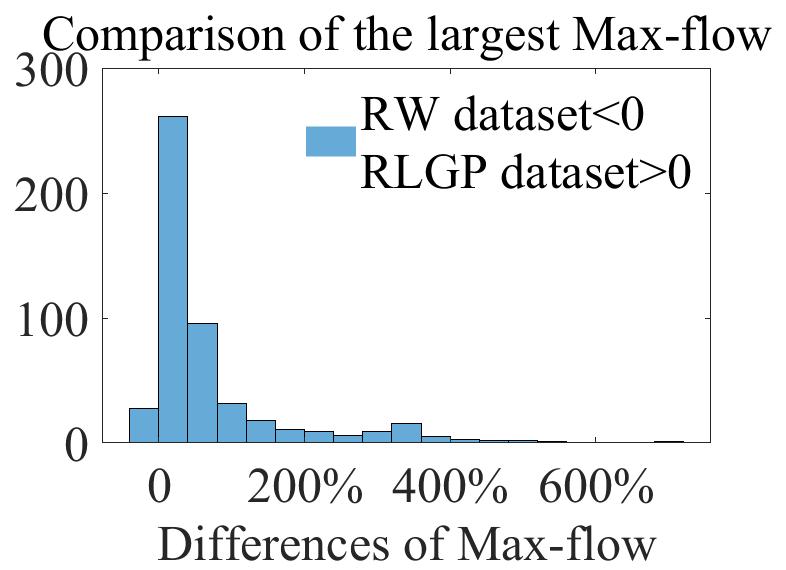}%
  \label{RLGP472_RW28_his}}
  \subfloat[]{\includegraphics[width=1.75in]{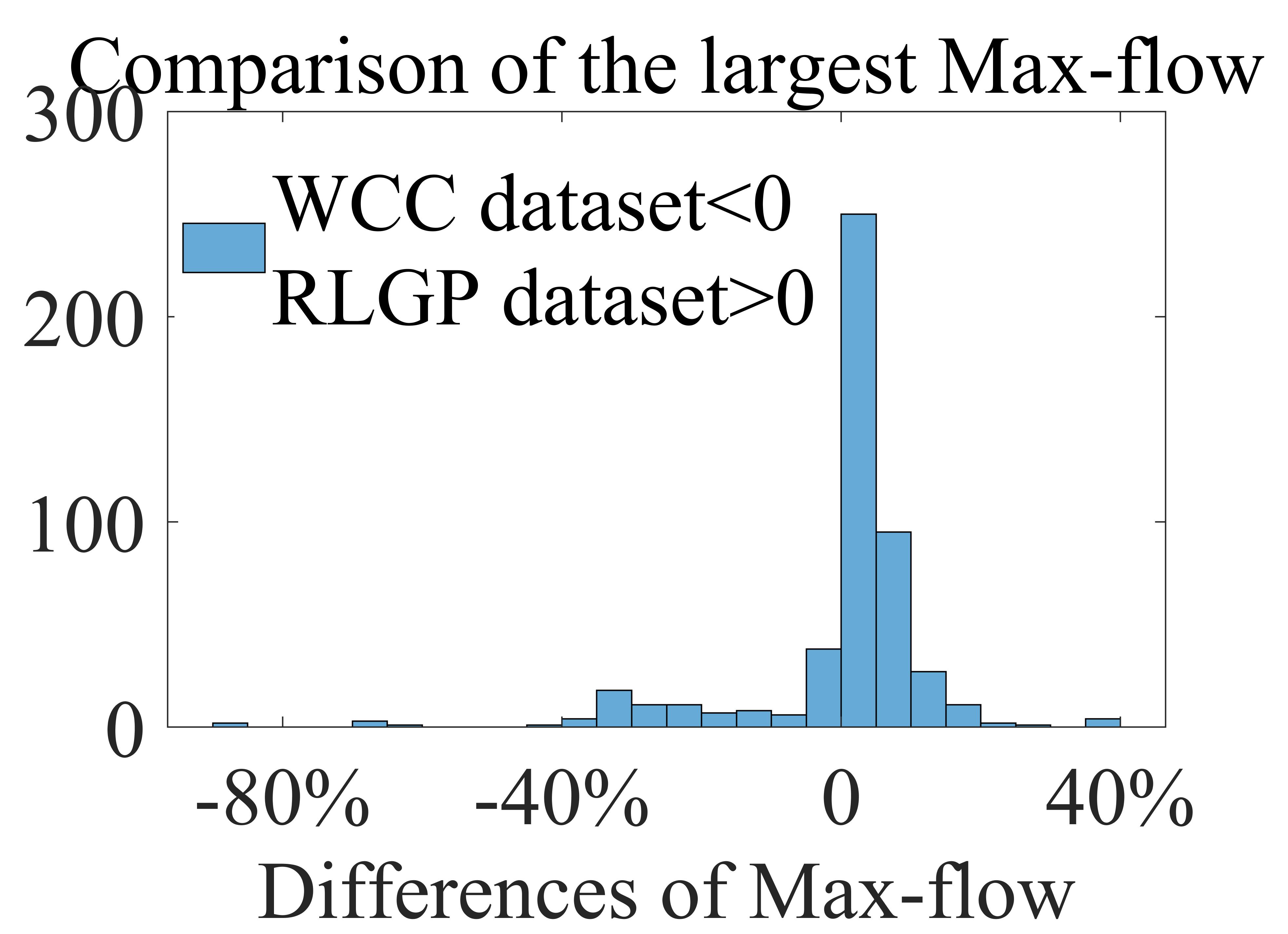}%
  \label{RLGP390_WCC110_his}}
  \caption{Comparison of different datasets for training GNN. (a) RLCP dataset vs. WCC dataset, (b) RLCP dataset vs. RW dataset.}
  \label{result_compar_diff_data}
\end{figure}

\subsubsection{Comparison of different graph convolution layers}\label{subsec:compar_diff_gcl}
To show the effect of using expressive GNN layers, we compare using GraphConv with GINEConv, which is the convolutional layer used in the graph isomorphism network \cite{xu2018how}, with edge features. While GraphConv is a higher-order method, GINEConv is first-order, and thus less expressive and has more deficient approximation power.

We apply both GraphConv and GINEConv to the GNN in the MFL method respectively. The comparison results are shown in Fig.~\ref{Graph478_GIN22_his}.
\begin{figure}[!t]
  \centering
  \includegraphics[width=1.75in]{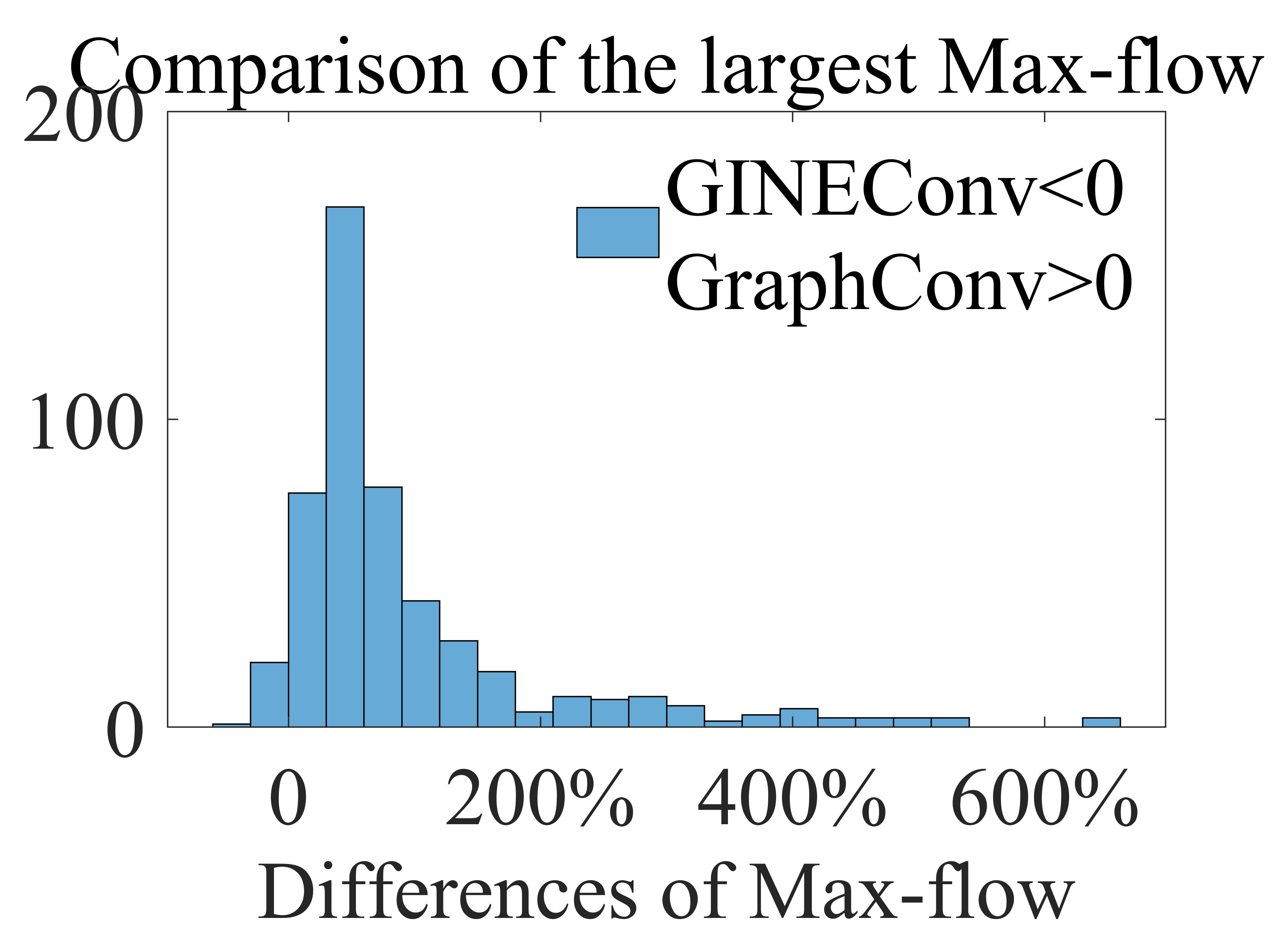}
  \caption{Comparison between GraphConv and GINEConv.}
  \label{Graph478_GIN22_his}
\end{figure}
As we can see, GNN with GraphConv layers performs much better than GINEConv.
$478$ (more than $95.6\%$) out of $500$  deployments see a superior performance of GNN with GraphConv. In particular, in $402$ of these deployments, GraphConv achieves a final max-flow at least $30\%$ larger than GINEConv.

\subsection{Additional Results on Accuracy of Max-flow Learning}\label{subsec:comparison_actual_max_flow}
In addition to the above experimental results, we compare the learned max-flow with the true values. 
We take $500$ deployments where the network nodes are located at the ending positions determined by the hybrid method on the test dataset. We show the average relative error (avg. rel. err.) and the truncated-averaged relative error (c-avg. rel. err.) excluding the highest and lowest $1\%$ values in Table~\ref{table:diff_approx}. It can be seen that the GNN-based MFL can achieve a more precise approximation of the max-flow, which is significantly superior to the WCC method.

\begin{table}[!t]
  \caption{Comparison of the average error to the actual max-flow}
  \label{table:diff_approx}
  \centering
  \renewcommand\arraystretch{1.5}
  \begin{tabular}{|c||c|c|}
    \hline
        ~ & avg. rel. err. & c-avg. rel. err.     \\ 
    \hline
\tabincell{c}{$\lambda_{2}(\mL_{\mW}) $ }   & 93.82\%  & 93.86\% \\
\hline
    \tabincell{c}{upper bound of WCC }   & 16.03\%  &14.03\%  \\ 
    \hline
    \tabincell{c}{lower bound of WCC  }  &  96.91\% &  96.93\%                                                     \\ 
    \hline
  \tabincell{c}{output of MFL}  & \bf{1.55\%}   &  \bf{ 0.79}\% \\ 
    \hline
\end{tabular}
\end{table}

\section{Conclusion}\label{sec:conclusion}
In this paper, we studied the problem of throughput maximization by optimizing the node deployment of wireless communication networks.
Specifically, we modeled the throughput as a network max-flow and used an expressive high-order GNN to learn the max-flow value under the various deployments of networks in a supervised fashion.
In the test phase, the trained GNN provided gradients for updating the locations of the relay nodes.
The correctness and effectiveness of the proposed approach were supported theoretically by investigating the approximation capabilities of permutation-invariant neural networks.
By searching the effective dataset with reinforcement learning, we  achieved competitive results compared with baselines.
Combining the existing spectral graph theory approach and the GNN-based approach, the proposed hybrid method provided some improvement in performance.
We believe that our method has a wide spectrum of potential applications such as unmanned aerial vehicle communication networks, mobile ad hoc networks,and robotic sensor networks. 

There are some limitations to the current method. In particular, this GNN-based method does not perform as well as other methods in some special deployments, which also cannot be completely solved by the hybrid method.
We hypothesize the reason is that in those deployments, the max-flow, as a function of locations, is not smooth enough and the generalizability of the underlying GNN model will deteriorate.
In the future, we will look for better approximation models to handle these cases.

\nocite*
\bibliographystyle{IEEEtran}
\bibliography{IEEEabrv, ref2}

\begin{thebibliography}{10}
\providecommand{\url}[1]{#1}
\csname url@samestyle\endcsname
\providecommand{\newblock}{\relax}
\providecommand{\bibinfo}[2]{#2}
\providecommand{\BIBentrySTDinterwordspacing}{\spaceskip=0pt\relax}
\providecommand{\BIBentryALTinterwordstretchfactor}{4}
\providecommand{\BIBentryALTinterwordspacing}{\spaceskip=\fontdimen2\font plus
\BIBentryALTinterwordstretchfactor\fontdimen3\font minus
  \fontdimen4\font\relax}
\providecommand{\BIBforeignlanguage}[2]{{%
\expandafter\ifx\csname l@#1\endcsname\relax
\typeout{** WARNING: IEEEtran.bst: No hyphenation pattern has been}%
\typeout{** loaded for the language `#1'. Using the pattern for}%
\typeout{** the default language instead.}%
\else
\language=\csname l@#1\endcsname
\fi
#2}}
\providecommand{\BIBdecl}{\relax}
\BIBdecl

\bibitem{chin2014emerging}
W.~H. Chin, Z.~Fan, and R.~Haines, ``Emerging technologies and research
  challenges for {5G} wireless networks,'' \emph{IEEE Wireless Commun.},
  vol.~21, no.~2, pp. 106--112, 2014.

\bibitem{gupta2015sinr}
A.~K. Gupta, X.~Zhang, and J.~G. Andrews, ``{SINR} and throughput scaling in
  ultradense urban cellular networks,'' \emph{{IEEE} Wireless Commun. Lett.},
  vol.~4, no.~6, pp. 605--608, 2015.

\bibitem{wu2012cost}
T.-J. Wu, W.~Liao, and C.-J. Chang, ``A cost-effective strategy for road-side
  unit placement in vehicular networks,'' \emph{IEEE Trans. Commun.}, vol.~60,
  no.~8, pp. 2295--2303, 2012.

\bibitem{you20206g}
X.~You, H.~Yin, and H.~Wu, ``On {6G} and wide-area {IoT},'' \emph{Chin. J.
  Internet Things}, vol.~4, no.~1, pp. 3--11, 2020.

\bibitem{merwaday2015uav}
A.~Merwaday and I.~Guvenc, ``{UAV} assisted heterogeneous networks for public
  safety communications,'' in \emph{Proc. IEEE Wireless Commun. Netw. Conf.
  (WCNC)}, 2015, pp. 329--334.

\bibitem{wang2014coverage}
H.~Wang, X.~Zhou, and M.~C. Reed, ``Coverage and throughput analysis with a
  non-uniform small cell deployment,'' \emph{IEEE Trans. Wireless Commun.},
  vol.~13, no.~4, pp. 2047--2059, 2014.

\bibitem{sung2013attainable}
D.~H. Kang, K.~W. Sung, and J.~Zander, ``Attainable user throughput by dense
  {Wi-Fi} deployment at 5 {GHz},'' in \emph{Proc. IEEE PIMRC}, 2013, pp.
  3418--3422.

\bibitem{lyu2020spatial}
J.~Lyu and R.~Zhang, ``Spatial throughput characterization for intelligent
  reflecting surface aided multiuser system,'' \emph{{IEEE} Wireless Commun.
  Lett.}, vol.~9, no.~6, pp. 834--838, 2020.

\bibitem{wang2018trajectory}
H.~Wang, J.~Chen, G.~Ding, and J.~Sun, ``Trajectory planning in {UAV}
  communication with jamming,'' in \emph{Proc. Int. Conf. Wireless Commun.
  Signal Process. (WCSP)}, 2018, pp. 1--6.

\bibitem{chou2019energy}
S.-F. Chou, A.-C. Pang, and Y.-J. Yu, ``Energy-aware {3D} unmanned aerial
  vehicle deployment for network throughput optimization,'' \emph{IEEE Trans.
  Wireless Commun.}, vol.~19, no.~1, pp. 563--578, 2019.

\bibitem{he2014dynamic}
X.~He, H.~Dai, and P.~Ning, ``Dynamic adaptive anti-jamming via controlled
  mobility,'' \emph{IEEE Trans. Wireless Commun.}, vol.~13, no.~8, pp.
  4374--4388, 2014.

\bibitem{rahmati2021dynamic}
A.~Rahmati, S.~Hosseinalipour, Y.~Yapici, X.~He, I.~Guvenc, H.~Dai, and
  A.~Bhuyan, ``Dynamic interference management for {UAV-assisted} wireless
  networks,'' \emph{IEEE Trans. Wireless Commun.}, 2021.

\bibitem{yang2018three}
P.~Yang, X.~Cao, X.~Xi, Z.~Xiao, and D.~Wu, ``Three-dimensional drone-cell
  deployment for congestion mitigation in cellular networks,'' \emph{IEEE
  Trans. Veh. Technol.}, vol.~67, no.~10, pp. 9867--9881, 2018.

\bibitem{mozaffari2016efficient}
M.~Mozaffari, W.~Saad, M.~Bennis, and M.~Debbah, ``Efficient deployment of
  multiple unmanned aerial vehicles for optimal wireless coverage,'' \emph{IEEE
  Commun. Lett.}, vol.~20, no.~8, pp. 1647--1650, 2016.

\bibitem{chen2017caching}
M.~Chen, M.~Mozaffari, W.~Saad, C.~Yin, M.~Debbah, and C.~S. Hong, ``Caching in
  the sky: Proactive deployment of cache-enabled unmanned aerial vehicles for
  optimized quality-of-experience,'' \emph{IEEE J. Select. Areas Commun.},
  vol.~35, no.~5, pp. 1046--1061, 2017.

\bibitem{wu2020comprehensive}
Z.~Wu, S.~Pan, F.~Chen, G.~Long, C.~Zhang, and S.~Y. Philip, ``A comprehensive
  survey on graph neural networks,'' \emph{IEEE Trans. Neural Netw. Learn.
  Syst.}, vol.~32, no.~1, pp. 4--24, 2020.

\bibitem{abadal2021computing}
S.~Abadal, A.~Jain, R.~Guirado, J.~L{\'o}pez-Alonso, and E.~Alarc{\'o}n,
  ``Computing graph neural networks: A survey from algorithms to
  accelerators,'' \emph{ACM Comput. Surv.}, vol.~54, no.~9, pp. 1--38, 2021.

\bibitem{zhang2021scalable}
X.~Zhang, H.~Zhao, J.~Xiong, X.~Liu, L.~Zhou, and J.~Wei, ``Scalable power
  control/beamforming in heterogeneous wireless networks with graph neural
  networks,'' in \emph{Proc. IEEE Glob. Commun. Conf.}, 2021, pp. 01--06.

\bibitem{eisen2020optimal}
M.~Eisen and A.~Ribeiro, ``Optimal wireless resource allocation with random
  edge graph neural networks,'' \emph{IEEE Trans. Signal Process.}, vol.~68,
  pp. 2977--2991, 2020.

\bibitem{guo2021learning}
J.~Guo and C.~Yang, ``Learning power control for cellular systems with
  heterogeneous graph neural network,'' in \emph{Proc. IEEE Wireless Commun.
  Netw. Conf.}, 2021, pp. 1--6.

\bibitem{orhan2021connection}
O.~Orhan, V.~N. Swamy, T.~Tetzlaff, M.~Nassar, H.~Nikopour, and S.~Talwar,
  ``Connection management {xAPP} for {O-RAN RIC}: A graph neural network and
  reinforcement learning approach,'' in \emph{Proc. IEEE Int. Conf. Mach.
  Learn. Appl.}, 2021, pp. 936--941.

\bibitem{moon2021neuro}
S.~Moon, S.~Ahn, K.~Son, J.~Park, and Y.~Yi, ``{Neuro-DCF: Design of Wireless
  MAC via Multi-Agent Reinforcement Learning Approach},'' in \emph{Proc. ACM
  Mobihoc}, 2021, pp. 141--150.

\bibitem{bruna2013spectral}
J.~Bruna, W.~Zaremba, A.~Szlam, and Y.~LeCun, ``Spectral networks and locally
  connected networks on graphs,'' in \emph{Proc. Int. Conf. Learn. Represent.},
  2020.

\bibitem{kondor2018covariant}
R.~Kondor, T.~S. Hy, H.~Pan, B.~M. Anderson, and S.~Trivedi, ``Covariant
  compositional networks for learning graphs,'' in \emph{Proc. Int. Conf.
  Learn. Represent. Workshop}, 2018.

\bibitem{zou2020graph}
D.~Zou and G.~Lerman, ``Graph convolutional neural networks via scattering,''
  \emph{Appl. Comput. Harmon. Anal.}, vol.~49, no.~3, pp. 1046--1074, 2020.

\bibitem{keriven2019universal}
N.~Keriven and G.~Peyr{\'e}, ``Universal invariant and equivariant graph neural
  networks,'' \emph{Proc. NeurIPS}, vol.~32, 2019.

\bibitem{maron2018invariant}
\BIBentryALTinterwordspacing
H.~Maron, H.~Ben-Hamu, N.~Shamir, and Y.~Lipman, ``Invariant and equivariant
  graph networks,'' \emph{arXiv:1812.09902.}, 2018. [Online]. Available:
  \url{https://arxiv.org/abs/1812.09902}
\BIBentrySTDinterwordspacing

\bibitem{hornik1991approximation}
K.~Hornik, ``Approximation capabilities of multilayer feedforward networks,''
  \emph{Neural Netw}, vol.~4, no.~2, pp. 251--257, 1991.

\bibitem{chen2019equivalence}
Z.~Chen, S.~Villar, L.~Chen, and J.~Bruna, ``On the equivalence between graph
  isomorphism testing and function approximation with gnns,'' in \emph{Proc.
  NeurIPS}, 2019, pp. 15\,894--15\,902.

\bibitem{xu2018how}
K.~Xu, W.~Hu, J.~Leskovec, and S.~Jegelka, ``How powerful are graph neural
  networks?'' in \emph{Proc. Int. Conf. Learn. Represent.}, 2019.

\bibitem{morris2019weisfeiler}
C.~Morris, M.~Ritzert, M.~Fey, W.~L. Hamilton, J.~E. Lenssen, G.~Rattan, and
  M.~Grohe, ``Weisfeiler and leman go neural: Higher-order graph neural
  networks,'' in \emph{Proc. AAAI Conf. Artif. Intell.}, vol.~33, no.~01, 2019,
  pp. 4602--4609.

\bibitem{bodnar2021weisfeiler}
C.~Bodnar, F.~Frasca, Y.~Wang, N.~Otter, G.~F. Montufar, P.~Lio, and
  M.~Bronstein, ``Weisfeiler and lehman go topological: Message passing
  simplicial networks,'' in \emph{Proc. Int. Conf. Mach. Learn.}, 2021, pp.
  1026--1037.

\bibitem{maron2019universality}
H.~Maron, E.~Fetaya, N.~Segol, and Y.~Lipman, ``On the universality of
  invariant networks,'' in \emph{Proc. Int. Conf. Mach. Learn.}, 2019, pp.
  4363--4371.

\bibitem{nguyen1999approximation}
T.~Nguyen-Thien and T.~Tran-Cong, ``Approximation of functions and their
  derivatives: A neural network implementation with applications,'' \emph{Appl.
  Math. Model.}, vol.~23, no.~9, pp. 687--704, 1999.

\bibitem{czarnecki2017sobolev}
W.~M. Czarnecki, S.~Osindero, M.~Jaderberg, G.~Swirszcz, and R.~Pascanu,
  ``Sobolev training for neural networks,'' in \emph{Proc. NeurIPS}, 2017, pp.
  4281--4290.

\bibitem{vlassis2021sobolev}
N.~N. Vlassis and W.~Sun, ``Sobolev training of thermodynamic-informed neural
  networks for interpretable elasto-plasticity models with level set
  hardening,'' \emph{Comput. Methods Appl. Mech. Eng.}, vol. 377, p. 113695,
  2021.

\bibitem{son2021sobolev}
\BIBentryALTinterwordspacing
H.~Son, J.~W. Jang, W.~J. Han, and H.~J. Hwang, ``Sobolev training for the
  neural network solutions of pdes,'' \emph{arXiv:2101.08932.}, 2021. [Online].
  Available: \url{https://arxiv.org/abs/2101.08932}
\BIBentrySTDinterwordspacing

\bibitem{ford1956maximal}
L.~R. Ford and D.~R. Fulkerson, ``Maximal flow through a network,'' \emph{Can.
  J. Math.}, vol.~8, pp. 399--404, 1956.

\bibitem{kim2021lipschitz}
H.~Kim, G.~Papamakarios, and A.~Mnih, ``The lipschitz constant of
  self-attention,'' in \emph{Proc. Int. Conf. Mach. Learn.}, 2021, pp.
  5562--5571.

\bibitem{federer2014geometric}
H.~Federer, \emph{Geometric measure theory}.\hskip 1em plus 0.5em minus
  0.4em\relax Springer, 2014.

\bibitem{hornik1989multilayer}
K.~Hornik, M.~Stinchcombe, and H.~White, ``Multilayer feedforward networks are
  universal approximators,'' \emph{Neural Netw}, vol.~2, no.~5, pp. 359--366,
  1989.

\bibitem{brenner2008mathematical}
S.~C. Brenner, L.~R. Scott, and L.~R. Scott, \emph{The mathematical theory of
  finite element methods}.\hskip 1em plus 0.5em minus 0.4em\relax Springer,
  2008, vol.~3.

\bibitem{he2020relu}
J.~He, L.~Li, J.~Xu, and C.~Zheng, ``{ReLU} deep neural networks and linear
  finite elements,'' \emph{J. Comput. Math.}, vol.~38, no.~3, pp. 502--527,
  2020.

\bibitem{azizian2021expressive}
W.~Azizian and marc lelarge, ``{Expressive Power of Invariant and Equivariant
  Graph Neural Networks},'' in \emph{Proc. Int. Conf. Learn. Represent.}, 2021.

\bibitem{jegelka2022theory}
\BIBentryALTinterwordspacing
S.~Jegelka, ``{Theory of Graph Neural Networks: Representation and Learning},''
  \emph{arXiv:2204.07697.}, 2022. [Online]. Available:
  \url{https://arxiv.org/abs/2204.07697}
\BIBentrySTDinterwordspacing

\bibitem{hendrycks2016gaussian}
\BIBentryALTinterwordspacing
D.~Hendrycks and K.~Gimpel, ``Gaussian error linear units (gelus),''
  \emph{arXiv:1606.08415.}, 2016. [Online]. Available:
  \url{https://arxiv.org/abs/1606.08415}
\BIBentrySTDinterwordspacing

\bibitem{dwivedi2020benchmarking}
\BIBentryALTinterwordspacing
V.~P. Dwivedi, C.~K. Joshi, T.~Laurent, Y.~Bengio, and X.~Bresson,
  ``Benchmarking graph neural networks,'' \emph{arXiv:2003.00982.}, 2020.
  [Online]. Available: \url{https://arxiv.org/abs/2003.00982}
\BIBentrySTDinterwordspacing

\bibitem{sherman2009breaking}
J.~Sherman, ``Breaking the multicommodity flow barrier for {O} ($\sqrt{ log
  n}$)-approximations to sparsest cut,'' in \emph{Proc. 50th Annu. IEEE Symp.
  Found. Comput. Sci. (FOCS)}, 2009, pp. 363--372.

\bibitem{bhattacharya2010graph}
S.~Bhattacharya and T.~Ba{\c{s}}ar, ``Graph-theoretic approach for connectivity
  maintenance in mobile networks in the presence of a jammer,'' in \emph{Proc.
  IEEE CDC}, 2010, pp. 3560--3565.

\bibitem{schulman2017proximal}
\BIBentryALTinterwordspacing
J.~Schulman, F.~Wolski, P.~Dhariwal, A.~Radford, and O.~Klimov, ``Proximal
  policy optimization algorithms,'' \emph{arXiv:1707.06347.}, 2017. [Online].
  Available: \url{https://arxiv.org/abs/1707.06347}
\BIBentrySTDinterwordspacing

\bibitem{heess2017emergence}
\BIBentryALTinterwordspacing
N.~Heess, D.~TB, S.~Sriram, J.~Lemmon, J.~Merel, G.~Wayne, Y.~Tassa, T.~Erez,
  Z.~Wang, S.~M.~A. Eslami, M.~A. Riedmiller, and D.~Silver, ``Emergence of
  locomotion behaviours in rich environments,'' \emph{arXiv:1707.02286.}, 2017.
  [Online]. Available: \url{https://arxiv.org/abs/1707.02286}
\BIBentrySTDinterwordspacing

\bibitem{valiulahi2020multi}
I.~Valiulahi and C.~Masouros, ``{Multi-UAV} deployment for throughput
  maximization in the presence of co-channel interference,'' \emph{{IEEE}
  Internet Things J.}, vol.~8, no.~5, pp. 3605--3618, 2020.

\bibitem{goldenberg2004towards}
D.~K. Goldenberg, J.~Lin, A.~S. Morse, B.~E. Rosen, and Y.~R. Yang, ``Towards
  mobility as a network control primitive,'' in \emph{Proc. ACM Mobihoc}, 2004,
  pp. 163--174.

\bibitem{fey2019fast}
M.~Fey and J.~E. Lenssen, ``Fast graph representation learning with {PyTorch
  Geometric},'' in \emph{Proc. ICLR Workshop Representation Learn. Graphs
  Manifolds}, 2019.

\end{thebibliography}

\end{document}